\documentclass{article}
\usepackage{amsmath}
\usepackage{amssymb}\usepackage{amsthm}
\usepackage{amsfonts}
\begin{document}
\theoremstyle{plain}
\newtheorem{thm}{Theorem}[section]
\newtheorem{lem}[thm]{Lemma}
\newtheorem{prop}[thm]{Proposition}
\newtheorem{cor}[thm]{Corollary}
\theoremstyle{definition}
\newtheorem{assum}[thm]{Assumption}
\newtheorem{notation}[thm]{Notation}
\newtheorem{defn}[thm]{Definition}
\newtheorem{clm}[thm]{Claim}
\newtheorem{ex}[thm]{Example}
\theoremstyle{remark}
\newtheorem{rem}[thm]{Remark}
\newcommand{\unit}{\mathbb I}
\newcommand{\ali}[1]{{\mathfrak A}_{[ #1 ,\infty)}}
\newcommand{\alm}[1]{{\mathfrak A}_{(-\infty, #1 ]}}
\newcommand{\nn}[1]{\lV #1 \rV}
\newcommand{\br}{{\mathbb R}}
\newcommand{\dm}{{\rm dom}\mu}
\newcommand{\lb}{l_{\bb}(n,n_0,k_R,k_L,\lal,\bbD,\bbG,Y)}
\newcommand{\Ad}{\mathop{\mathrm{Ad}}\nolimits}
\newcommand{\Proj}{\mathop{\mathrm{Proj}}\nolimits}
\newcommand{\RRe}{\mathop{\mathrm{Re}}\nolimits}
\newcommand{\RIm}{\mathop{\mathrm{Im}}\nolimits}
\newcommand{\Wo}{\mathop{\mathrm{Wo}}\nolimits}
\newcommand{\Prim}{\mathop{\mathrm{Prim}_1}\nolimits}
\newcommand{\Primz}{\mathop{\mathrm{Prim}}\nolimits}
\newcommand{\ClassA}{\mathop{\mathrm{ClassA}}\nolimits}
\newcommand{\Class}{\mathop{\mathrm{Class}}\nolimits}
\newcommand{\diam}{\mathop{\mathrm{diam}}\nolimits}
\def\qed{{\unskip\nobreak\hfil\penalty50
\hskip2em\hbox{}\nobreak\hfil$\square$
\parfillskip=0pt \finalhyphendemerits=0\par}\medskip}
\def\proof{\trivlist \item[\hskip \labelsep{\bf Proof.\ }]}
\def\endproof{\null\hfill\qed\endtrivlist\noindent}
\def\proofof[#1]{\trivlist \item[\hskip \labelsep{\bf Proof of #1.\ }]}
\def\endproofof{\null\hfill\qed\endtrivlist\noindent}
\newcommand{\caA}{{\mathcal A}}
\newcommand{\caB}{{\mathcal B}}
\newcommand{\caC}{{\mathcal C}}
\newcommand{\caD}{{\mathcal D}}
\newcommand{\caE}{{\mathcal E}}
\newcommand{\caF}{{\mathcal F}}
\newcommand{\caG}{{\mathcal G}}
\newcommand{\caH}{{\mathcal H}}
\newcommand{\caI}{{\mathcal I}}
\newcommand{\caJ}{{\mathcal J}}
\newcommand{\caK}{{\mathcal K}}
\newcommand{\caL}{{\mathcal L}}
\newcommand{\caM}{{\mathcal M}}
\newcommand{\caN}{{\mathcal N}}
\newcommand{\caO}{{\mathcal O}}
\newcommand{\caP}{{\mathcal P}}
\newcommand{\caQ}{{\mathcal Q}}
\newcommand{\caR}{{\mathcal R}}
\newcommand{\caS}{{\mathcal S}}
\newcommand{\caT}{{\mathcal T}}
\newcommand{\caU}{{\mathcal U}}
\newcommand{\caV}{{\mathcal V}}
\newcommand{\caW}{{\mathcal W}}
\newcommand{\caX}{{\mathcal X}}
\newcommand{\caY}{{\mathcal Y}}
\newcommand{\caZ}{{\mathcal Z}}
\newcommand{\bbA}{{\mathbb A}}
\newcommand{\bbB}{{\mathbb B}}
\newcommand{\bbC}{{\mathbb C}}
\newcommand{\bbD}{{\mathbb D}}
\newcommand{\bbE}{{\mathbb E}}
\newcommand{\bbF}{{\mathbb F}}
\newcommand{\bbG}{{\mathbb G}}
\newcommand{\bbH}{{\mathbb H}}
\newcommand{\bbI}{{\mathbb I}}
\newcommand{\bbJ}{{\mathbb J}}
\newcommand{\bbK}{{\mathbb K}}
\newcommand{\bbL}{{\mathbb L}}
\newcommand{\bbM}{{\mathbb M}}
\newcommand{\bbN}{{\mathbb N}}
\newcommand{\bbO}{{\mathbb O}}
\newcommand{\bbP}{{\mathbb P}}
\newcommand{\bbQ}{{\mathbb Q}}
\newcommand{\bbR}{{\mathbb R}}
\newcommand{\bbS}{{\mathbb S}}
\newcommand{\bbT}{{\mathbb T}}
\newcommand{\bbU}{{\mathbb U}}
\newcommand{\bbV}{{\mathbb V}}
\newcommand{\bbW}{{\mathbb W}}
\newcommand{\bbX}{{\mathbb X}}
\newcommand{\bbY}{{\mathbb Y}}
\newcommand{\bbZ}{{\mathbb Z}}
\newcommand{\str}{^*}
\newcommand{\lv}{\left \vert}
\newcommand{\rv}{\right \vert}
\newcommand{\lV}{\left \Vert}
\newcommand{\rV}{\right \Vert}
\newcommand{\la}{\left \langle}
\newcommand{\ra}{\right \rangle}
\newcommand{\ltm}{\left \{}
\newcommand{\rtm}{\right \}}
\newcommand{\lcm}{\left [}
\newcommand{\rcm}{\right ]}
\newcommand{\ket}[1]{\lv #1 \ra}
\newcommand{\bra}[1]{\la #1 \rv}
\newcommand{\lmk}{\left (}
\newcommand{\rmk}{\right )}
\newcommand{\al}{{\mathcal A}}
\newcommand{\md}{M_d({\mathbb C})}
\newcommand{\id}{\mathop{\mathrm{id}}\nolimits}
\newcommand{\Tr}{\mathop{\mathrm{Tr}}\nolimits}
\newcommand{\Ran}{\mathop{\mathrm{Ran}}\nolimits}
\newcommand{\Ker}{\mathop{\mathrm{Ker}}\nolimits}
\newcommand{\Aut}{\mathop{\mathrm{Aut}}\nolimits}
\newcommand{\spn}{\mathop{\mathrm{span}}\nolimits}
\newcommand{\Mat}{\mathop{\mathrm{M}}\nolimits}
\newcommand{\UT}{\mathop{\mathrm{UT}}\nolimits}
\newcommand{\DT}{\mathop{\mathrm{DT}}\nolimits}
\newcommand{\GL}{\mathop{\mathrm{GL}}\nolimits}
\newcommand{\spa}{\mathop{\mathrm{span}}\nolimits}
\newcommand{\supp}{\mathop{\mathrm{supp}}\nolimits}
\newcommand{\rank}{\mathop{\mathrm{rank}}\nolimits}
\newcommand{\idd}{\mathop{\mathrm{id}}\nolimits}
\newcommand{\ran}{\mathop{\mathrm{Ran}}\nolimits}
\newcommand{\dr}{ \mathop{\mathrm{d}_{{\mathbb R}^k}}\nolimits} 
\newcommand{\dc}{ \mathop{\mathrm{d}_{\cc}}\nolimits} \newcommand{\drr}{ \mathop{\mathrm{d}_{\rr}}\nolimits} 
\newcommand{\zin}{\mathbb{Z}}
\newcommand{\rr}{\mathbb{R}}
\newcommand{\cc}{\mathbb{C}}
\newcommand{\ww}{\mathbb{W}}
\newcommand{\nan}{\mathbb{N}}\newcommand{\bb}{\mathbb{B}}
\newcommand{\aaa}{\mathbb{A}}\newcommand{\ee}{\mathbb{E}}
\newcommand{\pp}{\mathbb{P}}
\newcommand{\wks}{\mathop{\mathrm{wk^*-}}\nolimits}
\newcommand{\he}{\hat {\mathbb E}}
\newcommand{\ikn}{{\caI}_{k,n}}
\newcommand{\mk}{{\Mat_k}}
\newcommand{\mnz}{\Mat_{n_0}}
\newcommand{\mn}{\Mat_{n}}
\newcommand{\mkk}{\Mat_{k_R+k_L+1}}
\newcommand{\mnzk}{\mnz\otimes \mkk}
\newcommand{\hbb}{H^{k,\bb}_{m,p,q}}
\newcommand{\gb}[1]{\Gamma^{(R)}_{#1,\bb}}
\newcommand{\cgv}[1]{\caG_{#1,\vv}}
\newcommand{\gv}[1]{\Gamma^{(R)}_{#1,\vv}}
\newcommand{\gvt}[1]{\Gamma^{(R)}_{#1,\vv(t)}}
\newcommand{\gbt}[1]{\Gamma^{(R)}_{#1,\bb(t)}}
\newcommand{\cgb}[1]{\caG_{#1,\bb}}
\newcommand{\cgbt}[1]{\caG_{#1,\bb(t)}}
\newcommand{\gvp}[1]{G_{#1,\vv}}
\newcommand{\gbp}[1]{G_{#1,\bb}}
\newcommand{\gbpt}[1]{G_{#1,\bb(t)}}
\newcommand{\Pbm}[1]{\Phi_{#1,\bb}}
\newcommand{\Pvm}[1]{\Phi_{#1,\bb}}
\newcommand{\mb}{m_{\bb}}
\newcommand{\E}[1]{\widehat{\mathbb{E}}^{(#1)}}
\newcommand{\lal}{{\boldsymbol\lambda}}
\newcommand{\rar}{{\boldsymbol r}}
\newcommand{\oo}{{\boldsymbol\omega}}
\newcommand{\vv}{{\boldsymbol v}}
\newcommand{\bbm}{{\boldsymbol m}}
\newcommand{\kl}[1]{{\mathcal K}_{#1}}
\newcommand{\wb}[1]{\widehat{B_{\mu^{(#1)}}}}
\newcommand{\ws}[1]{\widehat{\psi_{\mu^{(#1)}}}}
\newcommand{\wsn}[1]{\widehat{\psi_{\nu^{(#1)}}}}
\newcommand{\wv}[1]{\widehat{v_{\mu^{(#1)}}}}
\newcommand{\wbn}[1]{\widehat{B_{\nu^{(#1)}}}}
\newcommand{\wo}[1]{\widehat{\omega_{\mu^{(#1)}}}}
\newcommand{\dist}{\dc}
\newcommand{\hpu}{\hat P^{(n_0,k_R,k_L)}_R}
\newcommand{\hpd}{\hat P^{(n_0,k_R,k_L)}_L}
\newcommand{\pu}{ P^{(k_R,k_L)}_R}
\newcommand{\pd}{ P^{(k_R,k_L)}_L}
\newcommand{\puuz}{P_{R}^{(n_0-1,n_0-1)}\otimes P^{(k_R,k_L)}_R}
\newcommand{\pddz}{P_{L}^{(n_0-1,n_0-1)}\otimes P^{(k_R,k_L)}_L}
\newcommand{\puu}{\tilde P_R}
\newcommand{\pdd}{\tilde P_L}
\newcommand{\qu}[1]{ Q^{(k_R,k_L)}_{R, #1}}
\newcommand{\qd}[1]{ Q^{(k_R,k_L)}_{L,#1}}
\newcommand{\hqu}[1]{ \hat Q^{(n_0,k_R,k_L)}_{R, #1}}
\newcommand{\hqd}[1]{ \hat Q^{(n_0,k_R,k_L)}_{L,#1}}
\newcommand{\eij}[1] {E^{(k_R,k_L)}_{#1}}
\newcommand{\eijz}[1] {E^{(n_0-1,n_0-1)}_{#1}}
\newcommand{\heij}[1] {\hat E^{(k_R,k_L)}_{#1}}
\newcommand{\cn}{\mathop{\mathrm{CN}(n_0,k_R,k_L)}\nolimits}
\newcommand{\ghd}[1]{\mathop{\mathrm{GHL}(#1,n_0,k_R,k_L,\bbG)}\nolimits}
\newcommand{\ghu}[1]{\mathop{\mathrm{GHR}(#1,n_0,k_R,k_L,\bbD)}\nolimits}
\newcommand{\ghdb}[1]{\mathop{\mathrm{GHL}(#1,n_0,k_R,k_L,\bbG)}\nolimits}
\newcommand{\ghub}[1]{\mathop{\mathrm{GHR}(#1,n_0,k_R,k_L,\bbD)}\nolimits}
\newcommand{\hfu}[1]{{\mathfrak H}_{#1}^R}
\newcommand{\hfd}[1]{{\mathfrak H}_{#1}^L}
\newcommand{\hfui}[1]{{\mathfrak H}_{#1,1}^R}
\newcommand{\hfdi}[1]{{\mathfrak H}_{#1,1}^L}
\newcommand{\hfuz}[1]{{\mathfrak H}_{#1,0}^R}
\newcommand{\hfdz}[1]{{\mathfrak H}_{#1,0}^L}
\newcommand{\CN}{\overline{\hpd}\lmk\mnzk \rmk\overline{\hpu}}
\newcommand{\cnz}[1] {\chi_{#1}^{(n_0)}}
\newcommand{\eu}{\eta_{R}^{(k_R,k_L)}}
\newcommand{\ezu}{\eta_{R}^{(n_0-1,n_0-1)}}
\newcommand{\ed}{\eta_{L}^{(k_R,k_L)}}
\newcommand{\ezd}{\eta_{L}^{(n_0-1,n_0-1)}}
\newcommand{\fii}[1]{f_{#1}^{(k_R,k_L)}}
\newcommand{\fiir}[1]{f_{#1}^{(k_R,0)}}
\newcommand{\fiil}[1]{f_{#1}^{(0,k_L)}}
\newcommand{\fiz}[1]{f_{#1}^{(n_0-1,n_0-1)}}
\newcommand{\zeij}[1] {e_{#1}^{(n_0)}}
\newcommand{\CL}{\ClassA}
\newcommand{\CLn}{\Class_2(n,n_0,k_R,k_L)}
\newcommand{\braket}[2]{\left\langle#1,#2\right\rangle}
\newcommand{\abs}[1]{\left\vert#1\right\vert}
\newtheorem{nota}{Notation}[section]
\def\qed{{\unskip\nobreak\hfil\penalty50
\hskip2em\hbox{}\nobreak\hfil$\square$
\parfillskip=0pt \finalhyphendemerits=0\par}\medskip}
\def\proof{\trivlist \item[\hskip \labelsep{\bf Proof.\ }]}
\def\endproof{\null\hfill\qed\endtrivlist\noindent}
\def\proofof[#1]{\trivlist \item[\hskip \labelsep{\bf Proof of #1.\ }]}
\def\endproofof{\null\hfill\qed\endtrivlist\noindent}
\newcommand{\ZZ}{\bbZ_2\times\bbZ_2}
\newcommand{\SSS}{\mathcal{S}}
\newcommand{\cs}{s}
\newcommand{\ct}{t}
\newcommand{\hS}{S}
\newcommand{\new}[1]{#1}

\title{A $\bbZ_2$-index of symmetry protected topological phases with time reversal symmetry for quantum spin chains}
\author{Yoshiko Ogata \thanks{ Graduate School of Mathematical Sciences
The University of Tokyo, Komaba, Tokyo, 153-8914, Japan. 
Supported in part by
the Grants-in-Aid for
Scientific Research, JSPS.}}
\maketitle

\begin{abstract}
We introduce a $\bbZ_2$-index for time reversal invariant Hamiltonians with unique gapped ground state on quantum spin chains. 
We show this is an invariant of a $C^1$-classification of gapped Hamiltonians.
Analogous results hold
for more general on-site finite group symmetry, 
with the 2-cohomology class
as the invariant. 
\end{abstract}

\section{Introduction}
The notion of symmetry protected topological (SPT) phases was introduced by Gu and Wen [GW].
It is defined as follows:
we consider the set of all Hamiltonians with some symmetry, 
which have a unique gapped ground state in the bulk.
We regard two of such Hamiltonians are equivalent, if
there is a continuous path inside that family, connecting them.
By this equivalence relation, we may classify the Hamiltonians 
in this family.
A Hamiltonian which has only on-site interaction can be regarded as 
a trivial one. The set of Hamiltonians equivalent to such trivial ones represents a
trivial phase.
If a phase is nontrivial, it is a SPT phase.
A typical nontrivial example of an SPT phase is the Haldane phase \cite{Haldane1983a}\cite{Haldane1983b} in quantum spin chains with odd integer spin. 
Whether the Haldane phase is SPT or not has been studied substantially and
produced a fruitful theory of SPT phase 
\cite{Affleck:1988vr},
\cite{denNijsRommelse},\cite{Kennedy1990}
\cite{kt},\cite{kt2},\cite{Perez-Garcia2008},\cite{GuWen2009}, 
\cite{po},\cite{po2}, \cite{ChenGuWEn2011}.

A natural question to ask is what are invariants of this classification.
Following an earlier attempt in [PWSVC], 
Pollmann, Turner, Berg, and Oshikawa [PTBO1,PTBO2] introduced $\bbZ_2$-indices for injective matrix products states which have either $\bbZ_2\times \bbZ_2$ on-site symmetry(dihedral group of $\pi$-rotations about $x$, $y$, and $z$-axes), reflection symmetry, or time reversal symmetry.  
The $\bbZ_2$-index beyond the framework of matrix product state was recently introduced by Tasaki
for systems satisfying on-site $U(1)$-symmetry together with one of $\bbZ_2\times\bbZ_2$-onsite symmetry/reflection symmetry/time reversal symmetry \cite{ta}.
He showed that these are actually invariant of the classification.
In \cite{bn}, an operator called {\it excess spin} was introduced for two one-dimensional models with continuous symmetry,
and was shown to be related to the classification of gapped Hamiltonians on the half infinite chain.

In this paper, we focus on SPT phases of quantum spin chains in the bulk, with the time reversal symmetry.
We introduce a $\bbZ_2$-index for the time reversal invariant Hamiltonians with unique gapped ground state. The key ingredient is the projective representation associated to the unique bulk gapped ground state.
As the time reversal symmetry is discrete and anti-linear, something like excess spin looks hard to define. However, by considering the associated projective representation, we may define the $\bbZ_2$-index.
It turns out that this $\bbZ_2$-index is an invariant of the $C^1$-classification: suppose that
there is a $C^1$-path of interactions, and suppose that if we associate some suitable boundary condition,
it gives local Hamiltonians which are gapped for an
increasing sequence of finite boxes. (See Definition \ref{boundary}.)
Then the $\bbZ_2$-index does not change along the path.
These boundary conditions can be arbitrary, as long as they guarantee the gap. We may take it as periodic boundary condition, for example.
Furthermore, the boundary condition itself does not need to be time reversal invariant.
As stated above, for time reversal invariant injective matrix product state,
a $\bbZ_2$-index was introduced in \cite{po}.
It turns out that this is a special case of our $\bbZ_2$-index.

{\new { The key ingredient for the definition of the index is the {\it split property} that is satisfied by unique gapped ground states.
This important property of unique gapped ground states was proven by T. Matsui in \cite{Matsui2}.
Unfortunately, this significant result  has not been paid enough attention to so far 
(up to our recent paper \cite{ot} on  Lieb-Schultz-Mattis type theorems), because
it looked to have no new physical application. 
In this paper, we would like to emphasize that it {\it does}
have a very important application to a major question in physics, namely, to the classification problem of SPT phases.
The use of split property for classification problem is one important new suggestion from this paper, which turns out to be a very strong tool.
}}

Analogous argument can be carried out
for on-site finite group symmetry, in particular the $\bbZ_2\times\bbZ_2$-symmetry,
and the cohomology classs of projective representation 
is an invariant along the analogous path of such Hamiltonians. (See Appendix \ref{onsite}.)
The projective representation for on-site symmetry has been known for some time. In particular,  Matsui developed a mathematical theory for quantum spin chains based on projective representations of on-site group symmetry, for general pure states which satisfy the split property \cite{Matsui1}.

{\new {
What is new in this paper is identifying the 2-cohomology of the projective representation
as the index of SPT-phases, and showing that it is actually an invariant of the classification.
The projective representation shows up as follows.
Let us consider the unique gapped ground state
$\omega$ of an interaction satisfying the symmetry.
We consider the GNS triple 
$(\caH_R,\pi_R,\Omega_R)$
of the restriction $\left. \omega\right\vert_{\caA_R}$ of $\omega$ to the right infinite chain $\caA_R$.
The split property of $\omega$ is by definition that $\pi_R(\caA_R)''$
is $*$-isomorphic to $B(\caK)$, for some Hilbert space $\caK$.
By this $*$-isomorphism, the action of the symmetry on $\pi_R(\caA_R)''$
(which exists because of the invariance of $\omega$ under the symmetry) can be translated to
the action on $B(\caK)$.
However, by the Wigner Theorem, the action on $B(\caK)$ can be given by
unitary/ anti-unitary.
This gives the projective representation associated to $\omega$.
And our index is the 2-cohomology associated to this projective representation.

But how can one prove the stability of such objects like second cohomology class?
The index shows up rather abstractly in the context of the GNS
representation. In particular, it is not an observable.
Our idea is to carry out the quasi-equivalence argument
combined with the factorization property of quasi-local automorphisms.
More precisely, it is known \cite{bmns} that if $\Phi_0$ and $\Phi_1$ are interactions
connected by the path of gapped Hamiltonians, then there is
a quasi-local automorphism $\alpha$
which connects the corresponding unique gapped ground states as
$\omega_{\Phi_1}=\omega_{\Phi_0}\circ\alpha$.
We show that this $\alpha$ satisfies a factorization property, namely,
there are automorphisms $\alpha_R$ on $\caA_R$
and  $\alpha_L$ on $\caA_L$
such that $\alpha\circ \lmk \alpha_L^{-1}\otimes \alpha_R^{-1}\rmk$
is inner.
(When system satisfies the symmetry, $\alpha_R$ can be taken 
to commute with (anti-)automorphisms implementing the symmetry.)
From this and the split property,
we show that $\left.{\omega_{\Phi_0}}\right\vert_{\caA_R}\circ\alpha_R$
and ${\omega_{\Phi_1}}_{\caA_R}$
are quasi-equivalent.
This simple observation turns out to be the key for our analysis.
Using the $*$-isomorphism coming from this quasi-equivalence
and the $*$-isomorphisms coming from the split property
of $\omega_{\Phi_0}$, $\omega_{\Phi_1}$, we can prove the
stability of the index. (See section \ref{split}.)
Although the automorphic equivalence is a well-developed subject, 
we believe
the use of it combined with quasi-equivalence like this,
cannot be found in the literatures. This is a new technical suggestion from this paper, which is simple,
but turns out to be very useful.

}}

{\new{
We also would like to emphasize the importance of considering $\left. \omega\right\vert_{\caA_R}$
instead of $\omega$.
Of course we do have an action of the symmetry on the GNS triple
$(\caH,\pi,\Omega)$ of $\omega$.
But this action is given by a genuine unitary action $U$. It does not
give any nontrivial cohomology class.
One may argue that we still can investigate this genuine unitary representation $U$
as an invariant, like by investigating which irreducible representation is included in it. 
But in the paper \cite{OGS}, we show  that at least for
on-site linear action of finite group, {\it every} irreducible representation is contained
in $U$.
Because of this result, we do not expect that $U$ can play an important role in
the classification problem of SPT-phases.
We also would like to remark that importance of the cut is not specialized 
to our theory.
It is actually in the sprit of the original paper \cite{po} \cite{po2},
and of all the works considering the entanglement entropy. 
In the index theory of \cite{bdf}, it is also
very important to consider the half of the system.
}}
{\new{
The property of entanglement entropy can be addressed as well in our framework.
If the index is $-1$, 
the entanglement entropy is
larger than or equal to $\log 2$.
(See the end of Section \ref{index}.)
}}

The physical impact of our result is as follows.
In a word, we here develop a mathematical index theorem, with which the observations in the physics literature such as [CG,PTOB1,PTOB2,CGW] about the indices and phase structures in quantum spin chains with time reversal 
symmetry are 
made rigorous.
To see an important example, consider $S=1$ quantum spin chains with the AKLT interaction [AKLT]
\begin{equation}
\Phi_{\rm AKLT}(X)=
\begin{cases}
\sum_{\nu=1}^3S^\nu_jS^\nu_{j+1}+(\sum_{\nu=1}^3S^\nu_jS^\nu_{j+1})^2/3&\text{if $X=\{j,j+1\}$ with $j\in\bbZ$}\\
0&\text{otherwise},
\end{cases}
\end{equation}
and the trivial on-site intereaction
\begin{equation}
\Phi_{\rm trivial}(X)=
\begin{cases}
(S^3_j)^2&\text{if $X=\{j\}$ with $j\in\bbZ$}\\
0&\text{otherwise}.
\end{cases}
\end{equation}
See section \ref{index} for notations.
Both of them are time reversal invariant.
It is known that these models have a unique gapped ground state, which we denote by $\varphi_{\rm AKLT}$ and $\varphi_{\rm trivial}$. Both of $\varphi_{\rm AKLT}$ and $\varphi_{\rm trivial}$
are matrix product states.
{\new{What is believed in physics community is that 
AKLT and trivial interaction should belong to different phases, if we require time reversal symmetry
to be preserved.
For example, if we interpolate them as $s\Phi_{\rm AKLT}+(1-s)\Phi_{\rm trivial}$, with an interpolation parameter
$s\in[0,1]$, it is expected that there is a point $s\in (0,1)$
that the interaction does not have a unique gapped ground state.}}
Our theorem, along with results in [PTOB1,PTOB2,CGW,Tas2] about matrix product states, shows that these ground states are characterized by the indices $\sigma_{\varphi_{\rm AKLT}}=-1$ and $\sigma_{\varphi_{\rm trivial}}=1$.(See section \ref{oshikawa}, for our index in matrix product states.)
We then get the following.
\begin{cor}
The two interactions $\Phi_{\rm AKLT}$ and $\Phi_{\rm trivial}$ can never be connected by a $C^1$-path of time reversal invariant 
interactions satisfying the  {\it Condition B}.
\end{cor}
See section \ref{c1main} for the definition of {\it Condition B}.
Analogous result can be shown for $\ZZ$-symmetry.
This roughly means that there must be a gapless model in between the models with $\Phi_{\rm AKLT}$ and $\Phi_{\rm trivial}$, provided that the time reversal or $\ZZ$ on-site symmetry is preserved.\footnote{%
We show that, for any $C^1$-path of interactions with the required symmetry and for any $C^1$-path of boundary conditions without symmetry, there must be a point at which the energy gap of a finite chain vanishes as the length of the chain increases.
}
The conjuecture that the AKLT model is in a nontrivial SPT phase has been established for the cases with time reversal symmetry and on-site $\ZZ$ symmetry.
Note that [Tas1] proves similar result, but with extra assumption that the interactions are also U(1) invariant.
However, his path does not need to be $C^1$.

This article is organized as follows. In Section \ref{index}, we introduce the $\bbZ_2$-index
Actually, this index is defined not only for pure ground states of the gapped Hamiltonians, but more generally, it is defined for pure states satisfying the split property. By \cite{Matsui2}, it is known that
a pure ground state of the gapped Hamiltonian satisfies the split property.
We show that this $\bbZ_2$-index is the same for two pure split states,
if they are automorphic equivalent via a time reversal invariant automorphism
which allows a time reversal invariant factorization. (See Definition \ref{aet}, and Definieion \ref{adt}.)
The proof is given in Section \ref{split}.
For a unique gapped ground state of 
a time reversal invariant Hamiltonian, 
as it is a special type of a pure split state, we may associate the $\bbZ_2$-index.
It turns out that
for a path of gapped Hamiltonians given in Definition \ref{boundary},
two ends ground states are automorphic equivalent via a time reversal invariant automorphism
which allows a time reversal invariant factorization. (Proposition \ref{c1t}.)
Therefore, our $\bbZ_2$-index is an invariant of this $C^1$-classification.
The proof is given in Section \ref{c1}.
In Section \ref{oshikawa}, we show that
the $\bbZ_2$-index in \cite{po}
is a special case of our $\bbZ_2$-index.

\section{The $\bbZ_2$-index}\label{index}
We start by summarizing standard setup of quantum spin chains on the infinite chain \cite{BR1,BR2}.
Let $S$ be an element of $\frac 12 \bbN$ and let $\SSS=\{-S,-S+1,\ldots,S-1,S\}$.
We denote the algebra of $(2S+1)\times (2S+1)$ matrices by $\Mat_{2S+1}$.
We denote the standard basis of $\cc^{2S+1}$ by $\{\psi_\mu\}_{\mu\in\SSS}$, and 
set $e_{\mu,\nu}=\ket{\psi_\mu}\bra{\psi_\nu}$ for each $\mu,\nu\in\SSS$.
Let $\hS_1, \hS_2, \hS_3\in\Mat_{2S+1}$ be the standard spin operators 
on $\bbC^{2S+1}$.
They satisfy $(\hS_1)^2+(\hS_2)^2+(\hS_3)^2=S(S+1)$ and the commutation relations $[\hS_1,\hS_2]=i\hS_3$, $[\hS_2,\hS_3]=i\hS_1$, and $[\hS_3,\hS_1]=i\hS_2$. (See \cite{TasakiBook}.)

We denote the set of all finite subsets in ${\bbZ}$ by ${\mathfrak S}_{\bbZ}$,
and the set of all finite intervals in ${\bbZ}$ by ${\mathfrak I}_{\bbZ}$.
For each $X\in {\mathfrak S}_{\bbZ}$, $\diam(X)$ denotes the diameter of $X$.
For $X,Y\subset \bbZ$, we denote by $d(X,Y)$, the distance between them.
The number of elements in a finite set $\Lambda\subset {\bbZ}$ is denoted by
$|\Lambda|$. For each $n\in\bbN$, we denote $[-n,n]\cap \bbZ$ by $\Lambda_n$.
The complement of $\Lambda$ in $\bbZ$ is denoted by $\Lambda^c$.

For each $z\in\bbZ$,  let $\caA_{\{z\}}$ be an isomorphic copy of $\Mat_{2S+1}$, and for any finite subset $\Lambda\subset\bbZ$, let $\caA_{\Lambda} = \otimes_{z\in\Lambda}\caA_{\{z\}}$, which is the local algebra of observables in $\Lambda$. 
For finite $\Lambda$, the algebra $\caA_{\Lambda} $ can be regarded as the set of all bounded operators acting on
the Hilbert space $\otimes_{z\in\Lambda}{\bbC}^{2S+1}$.
We use this identification freely.
If $\Lambda_1\subset\Lambda_2$, the algebra $\caA_{\Lambda_1}$ is naturally embedded in $\caA_{\Lambda_2}$ by tensoring its elements with the identity. 
The algebra $\caA_{R}$ (resp. $\caA_L$) representing the half-infinite chain
is given as the inductive limit of the algebras $\caA_{\Lambda}$ with $\Lambda\in{\mathfrak S}_{\bbZ}$, $\Lambda\subset[0,\infty)$
(resp. $\Lambda\subset(-\infty -1]$). 
The algebra $\caA$, representing the two sided infinite chain
is given as the inductive limit of the algebras $\caA_{\Lambda}$ with $\Lambda\in{\mathfrak S}_{\bbZ}$. 
Note that $\caA_{\Lambda}$ for $\Lambda\in {\mathfrak S}_{\bbZ}$,
and $\caA_R$ can be regarded naturally as subalgebras of
$\caA$.
Under this identification, for each $z\in \bbZ$,
we denote the 
spin operators in $\caA_{\{z\}}\subset\caA$ by $\hS_{1}^{(z)}, \hS_{2}^{(z)}, \hS_{3}^{(z)}$.
We denote the set of local observables by $\caA_{\rm loc}=\bigcup_{\Lambda\in{\mathfrak S}_\bbZ}\caA_{\Lambda}
$.
We denote by $\beta_x$ the automorphisms on $\caA$ representing the space translation by  $x\in\bbZ$.

Time reversal is the unique  antilinear unital $*$-automorphism $\Xi$ on $\caA$ satisfying
\begin{align*}
\Xi(S_j^{(z)}) = -S_j^{(z)},\quad j=1,2,3,\quad z\in\bbZ.
\end{align*}
(See Appendix B of \cite{ot} for the existence of such an automorphism.)
Note that $\Xi$ commutes with $\beta_x$ for any $x\in\bbZ$.
As $\Xi(\caA_R)=\caA_R$ (resp. $\Xi(\caA_L)=\caA_L$), the restriction $\Xi_R:=\left. \Xi\right \vert_{\caA_R}$ (resp. $\Xi_L:=\left. \Xi\right \vert_{\caA_L}$)
is an antilinear unital $*$-automorphism on $\caA_R$ (resp. $\caA_L$). 
For a state $\varphi$ on $\caA$, its time reversal $\hat \varphi$ is
given by 
\[
\hat \varphi\lmk A\rmk
=\varphi\lmk\Xi\lmk A^*\rmk\rmk,\quad
A\in \caA.
\]
We say $\varphi$ is time reversal invariant if we have $\varphi=\hat\varphi$.
%
%

We introduce $\bbZ_2$-index for pure time reversal invariant state satisfying the split property.
Let us first recall the definition of the split property.
We here give the following definition of the split property, which is most suitable for our purpose.
It corresponds to the standard definition \cite{dl} in our setting (see \cite{Matsui2}).
\begin{defn}
Let $\varphi$ be a pure state on $\caA$. Let $\varphi_R$ be the restriciton of
$\varphi$ to $\caA_R$, and $(\caH_{\varphi_R},\pi_{\varphi_R},\Omega_{\varphi_R})$ be the GNS triple of $\varphi_R$.
We say $\varphi$ satisfies the split property with respect to $\caA_L$ and $\caA_R$,
if the von Neumann algebra $\pi_{\varphi_R}(\caA_{R})''$ is a type I factor.
\end{defn}
Recall that a type I factor is isomorphic to $B(\caK)$, the set of all bounded operators 
on a  Hilbert space $\caK$. 
See \cite{takesaki}.

\begin{thm}\label{sp}
Let $\varphi$ be a time reversal invariant pure state on $\caA$, which satisfies the split property. Let $\varphi_R$ be the restriciton of
$\varphi$ to $\caA_R$, and $(\caH_{\varphi_R},\pi_{\varphi_R},\Omega_{\varphi_R})$ be the GNS triple of $\varphi_R$.

Then there are a Hilbert space $\caK_\varphi$, a $*$-isomorphism $\iota_\varphi : \pi_{\varphi_R}\lmk \caA_R\rmk{''}\to B(\caK_{\varphi})$,
and an antiunitary $J_{\varphi}$ on $\caK_{\varphi}$ such that
\begin{align*}
\iota_\varphi\circ \pi_{\varphi_R}\circ\Xi_R\lmk A\rmk
=J_\varphi \lmk \iota_\varphi \circ\pi_{\varphi_R}\lmk A\rmk\rmk J_\varphi^*,\quad
A\in\caA_R.
\end{align*}
Futhermore, $J_\varphi^2=\sigma_\varphi\unit$, with $\sigma_\varphi=1$ or $\sigma_\varphi=-1$.
These $\caK_\varphi$, $\iota_\varphi$, $J_{\varphi}$ are unique in the following sense.: 
If  a Hilbert space  $\tilde \caK_\varphi$, a $*$-isomorphism $\tilde\iota_\varphi : \pi_{\varphi_R}\lmk \caA_R\rmk{''}\to B(\tilde \caK_{\varphi})$,
and an antiunitary $\tilde J_{\varphi}$ on $\tilde \caK_{\varphi}$ satisfy
\begin{align*}
\tilde \iota_\varphi\circ \pi_{\varphi_R}\circ\Xi_R\lmk A\rmk
=\tilde J_\varphi \lmk \tilde \iota_\varphi \circ\pi_{\varphi_R}\lmk A\rmk\rmk {\tilde J_\varphi}^*,\quad
A\in\caA_R,
\end{align*}
then
there is a unitary $W:\caK_\varphi\to \tilde \caK_\varphi$ and $e^{i\theta}\in \bbT$
such that
\begin{align*}
&W\lmk \iota_\varphi \lmk x\rmk\rmk W^*=
\tilde \iota_\varphi \lmk x\rmk,
 \quad x\in \pi_{\varphi_R}\lmk \caA_R\rmk{''},\\
&e^{i\theta} WJ_\varphi W^*=\tilde J_\varphi.
 \end{align*}
In particular, $\tilde {J_\varphi}^2=\sigma_\varphi\unit$.
\end{thm}
From this Theorem we may define the $\bbZ_2$-index for pure time reversal invariant state satisfying the split property.
\begin{defn}\label{indexsp}
From Theorem \ref{sp}, for each time reversal invariant pure state
$\varphi$ on $\caA$ with the split property,
we obtain a $\bbZ_2$-index associated to $\varphi$. We denote this $\bbZ_2$-index by $\sigma_\varphi\in\{-1,1\}$.
\end{defn}
We will prove Theorem \ref{sp} in section \ref{split}.
In section \ref{c1main}, we will see that we may associate this index to time reversal invariant models
with unique gapped ground state.
In section \ref{oshikawa}, we will prove that the index introduced by Pollmann et.al. \cite{po}
is the index in Definition \ref{indexsp}, in the special setting, i.e., for matrix product states.

Having an index, natural question to ask is if it is an invariant of some classification.
We show that this index is an invariant of the classification with respect to the {\it factorizable automorphic equivalence, preserving the time reversal symmetry}. First, let us define the automorphic equivalence preserving the time reversal symmetry.
\begin{defn}\label{aet}
Let $\varphi_1,\varphi_2$ be two time reversal invariant states on $\caA$.
We say $\varphi_2$ is automorphic equivalent to $\varphi_1$
via a time reversal invariant automorphism
if there exists an automorphism $\alpha$ on $\caA$ such that
\begin{align}
\varphi_2=\varphi_1\circ \alpha\quad\text{and}\quad
\alpha\circ\Xi=\Xi\circ\alpha.
\end{align}
\end{defn}
\begin{defn}\label{adt}
We say an automorphism $\alpha$ of $\caA$ is factorizable if there are automorphisms $\alpha_R$,
$\alpha_L$ on $\caA_R$, $\caA_L$ respectively, and a unitary $W$ in $\caA$
such that
\[
\alpha\circ\lmk\alpha_L^{-1}\otimes \alpha_R^{-1}\rmk(A)=
WAW^*,\quad A\in\caA.
\]
We call these $(\alpha_R,\alpha_L,W)$, a factorization of $\alpha$.
When $\alpha$ is time reversal invariant, i.e., $\Xi\circ \alpha=\alpha\circ\Xi$,
we say a factorization $(\alpha_R,\alpha_L,W)$ of $\alpha$
is time reversal invariant, if $\alpha_R\circ\Xi_R=\Xi_R\circ\alpha_R$ and
$\alpha_L\circ\Xi_L=\Xi_L\circ\alpha_L$.
When such factorization exists, we say that $\alpha$ allows a time reversal invariant factorization.
\end{defn}

\begin{thm}\label{aei}
Let $\varphi_1,\varphi_2$ be time reversal invariant pure states satisfying the split property.
Suppose that $\varphi_1$ and $\varphi_2$ 
are automorphic equivalent via a time reversal invariant automorphism
which allows a time reversal invariant factorization.
Then the $\bbZ_2$-indices $\sigma_{\varphi_1}$,  $\sigma_{\varphi_2}$associated to
$\varphi_1$, $\varphi_2$ (in Definition \ref{indexsp}) are equal.
\end{thm}
We will prove this theorem in section \ref{split}.
In the next section, we see that this theorem can be applied to the setting of $C^1$-classification of gapped Hamiltonians.
Therefore, the $\bbZ_2$-index is an invariant of the $C^1$-classification.
{\new{Any time reversal unique gapped ground states satisfy the conditions of Theorem \ref{aei}.
In particular, we may apply Theorem \ref{aei} to time reversal invariant matrix product states.
}}

{\new{
The non-trivial index implies a lower bound of the entanglement entropy.
The entanglement entropy in our setting is defined as follows.
Let $(\caH_{\varphi_R},\pi_{\varphi_R},\Omega_{\varphi_R})$, 
$\caK_\varphi$, $\iota_\varphi$, $J_{\varphi}$
be as in Theorem \ref{sp}.
Then there exists a density matrix $\rho$ on $\caK_\varphi$ such that
$\left. \omega\right\vert_{\caA_R}=\Tr_{\caK_\varphi}\lmk
\rho\iota_\varphi\circ \pi_{\varphi_R}\lmk \cdot\rmk\rmk$.
The entanglement entropy of $\varphi$ is the von Neumann entropy 
of $\rho$.
By the time reversal invariance of the system, we have
$J_{\varphi} \rho J_{\varphi}^*=\rho$.
From this, for each eigenvalue $\lambda$ of $\rho$,
the corresponding eigen space is $J_\varphi$-invariant.
When $\sigma_\varphi=-1$, any $J_\varphi$-invariant
subspace have to be of even dimension.
Therefore, any eigenvalue of $\rho$ are evenly degenerated.
This implies that the entanglement entropy  
of $\varphi$ is bounded from below by $\log 2$.
This corresponds to the bound of the entanglement entropy observed in \cite{po}.
}
}

\section{$C^1$-classification of gapped Hamiltonians with the time reversal symmetry.}\label{c1main}
Let us now apply the result in section \ref{index} to the $C^1$-classification of gapped Hamiltonians preserving the time reversal symmetry. 

A mathematical model of a quantum spin chain is fully specified by its interaction $\Phi$.
An interaction is a map $\Phi$ from 
${\mathfrak S}_{\bbZ}$ into ${\caA}_{\rm loc}$ such
that $\Phi(X) \in {\caA}_{X}$ 
and $\Phi(X) = \Phi(X)^*$
for $X \in {\mathfrak S}_{\bbZ}$. 
An interaction $\Phi$ is translation invariant if
$\Phi(X+x)=\beta_x(\Phi(X))$, 
for all $x\in{\mathbb Z}$ and $X\in  {\mathfrak S}_{\bbZ}$, and time reversal invariant
if $\Xi(\Phi(X))=\Phi(X)$
for all $X\in {\mathfrak S}_\bbZ$.
Furthermore, an interaction $\Phi$
 is of finite range if there exists an $m\in {\mathbb N}$ such that
$\Phi(X)=0$ for $X$ with diameter larger than $m$.
We denote by $\caB_{f}$, 
the set of all finite range interactions $\Phi$ which satisfy
\begin{align}\label{fi}
a_\Phi:= \sup_{X\in {\mathfrak S}_{\bbZ}}
\lV
\Phi\lmk X\rmk
\rV<\infty.
\end{align}
We may define addition on $\caB_f$: for $\Phi,\Psi\in\caB_f$,
$\Phi+\Psi$ denotes the interaction given by
$(\Phi+\Psi)(X)=\Phi(X)+\Psi(X)$ for each $X\in{\mathfrak S}_{\bbZ}$.

For an interaction $\Phi$ and a finite set $\Lambda\in{\mathfrak S}_{\bbZ}$, we define the local Hamiltonian on $\Lambda$ by
\begin{equation}\label{GenHamiltonian}
\lmk H_{\Phi}\rmk_{\Lambda}:=\sum_{X\subset{\Lambda}}\Phi(X).
\end{equation}
The dynamics given by this local Hamiltonian is denoted by
\begin{align}\label{taulamdef}
\tau_t^{\Phi,\Lambda}\lmk A\rmk:= e^{it\lmk H_{\Phi}\rmk_{\Lambda}} Ae^{-it\lmk H_{\Phi}\rmk_{\Lambda}},\quad
t\in\bbR.
\end{align}
If $\Phi$ belongs to $\caB_{f}$,
the limit
\begin{align}\label{taudef}
\tau_t^{\Phi}\lmk A\rmk=\lim_{\Lambda\to\bbZ}
\tau_t^{\Phi,\Lambda}\lmk A\rmk
\end{align}
exists for each $A\in \caA$ and $t\in{\mathbb R}$, 
and defines a strongly continuous one parameter group of automorphisms $\tau^\Phi$ on $\caA$. 
(See \cite{BR2}.)
We denote the generator of $C^* $-dynamics $\tau^{\Phi}$ by $\delta_{\Phi}$.


For $\Phi\in\caB_f$, a state $\varphi$ on $\caA$ is called a \mbox{$\tau^{\Phi}$-ground} state
if the inequality
$
-i\,\varphi(A^*{\delta_{\Phi}}(A))\ge 0
$
holds
for any element $A$ in the domain $\caD({\delta_{\Phi}})$ of ${\delta_\Phi}$.
Let $\varphi$ be a $\tau^\Phi$-ground state, with the GNS triple $(\caH_\varphi,\pi_\varphi,\Omega_\varphi)$.
Then there exists a unique positive operator $H_{\varphi,\Phi}$ on $\caH_\varphi$ such that
$e^{itH_{\varphi,\Phi}}\pi_\varphi(A)\Omega_\varphi=\pi_\varphi(\tau_t^\Phi(A))\Omega_\varphi$,
for all $A\in\caA$ and $t\in\mathbb R$.
We call this $H_{\varphi,\Phi}$ the bulk Hamiltonian associated with $\varphi$.
Note that $\Omega_\varphi$ is an eigenvector of $H_{\varphi,\Phi}$ with eigenvalue $0$. See \cite{BR2} for the general theory.

The following definition clarifies what we mean by a model with a unique gapped ground state.
\begin{defn}
We say that a model with an interaction $\Phi\in\caB_f$ has a unique gapped ground state if 
(i)~the $\tau^\Phi$-ground state, which we denote as $\varphi$, is unique, and 
(ii)~there exists a $\gamma>0$ such that
$\sigma(H_{\varphi,\Phi})\setminus\{0\}\subset [\gamma,\infty)$, where  $\sigma(H_{\varphi,\Phi})$ is the spectrum of $H_{\varphi,\Phi}$.
\end{defn}
Note that the uniqueness of $\varphi$ implies that 0 is a non-degenerate eigenvalue of $H_{\varphi,\Phi}$.

If $\varphi$ is a \mbox{$\tau^{\Phi}$-ground} state of time reversal invariant interaction $\Phi\in {\caB}_f$,
then its time reversal
$\hat \varphi$ is also a \mbox{$\tau^{\Phi}$-ground} state.
In particular, if $\varphi$ is a unique \mbox{$\tau^{\Phi}$-ground} state, it is pure and time reversal invariant.

In \cite{Matsui2}, T.Matsui showed that the spectral gap implies the split property.
\begin{thm}[Theorem 1.5, Lemma 4.1, and Proposition 4.2 of \cite{Matsui2}]\label{matsui}
Let $\varphi$ be a pure $\tau^\Phi$-ground state of $\Phi\in{\caB}_f$, and denote by  $H_{\varphi,\Phi}$ the corresponding bulk Hamiltonian. 
Assume that $0$ is a non-degenerate eigenvalue of $H_{\varphi,\Phi}$ and
there exists $\gamma>0$ such that $\sigma(H_{\varphi,\Phi})\setminus\{0\}\subset [\gamma,\infty)$.
Then $\varphi$ satisfies the  split property with respect to $\caA_L$ and $\caA_R$.
\end{thm}
%
This theorem, combined with Theorem \ref{sp} allows us to define the following $\bbZ_2$-index for
time reversal invariant Hamiltonians with unique gapped ground state.
\begin{defn}
Let $\Phi\in\caB_f$ be a time reversal invariant interaction
which has a unique gapped ground state $\varphi$.
By Theorem \ref{matsui}, $\varphi$ satisfies the split property. Hence we obtain the $\bbZ_2$-index $\sigma_\varphi$
in Definition  \ref{indexsp}.
In this setting, we denote this $\sigma_\varphi$
 by $\hat \sigma_{\Phi}$ and call it the $\bbZ_2$-index associated to $\Phi$.
\end{defn}


As $\hat \sigma_\Phi$ takes discrete values $\{-1,1\}$,
for a continuous path of interactions $\Phi(s)$, we would expect that
$\hat \sigma_{\Phi(s)}$ is constant.
We prove this in the setting of $C^1$-classsification.
\begin{defn}\label{boundary}
We say the map $\Phi:[0,1]\ni s \to \Phi(s):=\{\Phi(X;s)\}_{X\in {\mathfrak S}_\bbZ}\in {\caB_f}$ is a $C^1$-path of time reversal invariant gapped interactions satisfying the {\it Condition B},
if 
there exist 
\begin{description}
\item[(i)] numbers
$M,R\in\nan$, $\gamma>0$  and an increasing sequence $n_k\in\nan$, $k=1,2,\ldots$,
\item[(ii)] $C^1$-functions
$a,b:[0,1]\to \bbR$ such that $a(s)<b(s)$,
\item[(iii)]
a sequence of paths of interactions $\Psi_k:[0,1]\ni s \to \Psi_k(s):=\{\Psi_k(X;s)\}_{X\in {\mathfrak S}_\bbZ}\in {\caB_f}$, $k=1,2,\ldots$,
\end{description}
and the following hold.
\begin{enumerate}
\item For each $X\in{\mathfrak S}_\bbZ$, the map
$[0,1]\ni s\to \Phi(X;s), \Psi_k(X;s)\in\caA_{X}$ are continuous and piecewise $C^1$.
We denote by $\Phi'(X;s)$, $\Psi'_k(X;s)$, the corresponding derivatives.
\item For each $s\in[0,1]$, and $X\in{\mathfrak S}_\bbZ$ with $\diam (X)\ge M$, we have
$\Phi(X;s)=0$.
\item For each $s\in[0,1]$, and $k\in\bbN$, we have $\Psi_k(X;s)=0$ unless $X\subset \Lambda_{n_k}\setminus \Lambda_{n_k-R}$.
\item Interactions are bounded as follows
\begin{align}
C_1:=\sup_{s\in[0,1]}\sup_{k\in\nan}\sup_{X\in {\mathfrak S}_\bbZ}
\lmk
\lV
\Phi\lmk X;s\rmk
\rV+|X|\lV
\Phi' \lmk X;s\rmk
\rV
+
\lV
\Psi_k\lmk X;s\rmk
\rV+|X|\lV
\Psi_k'\lmk X;s\rmk
\rV
\rmk<\infty.
\end{align}
\item For each $s\in[0,1]$, there exists a unique $\tau^{\Phi(s)}$-ground state
$\varphi_s$. 
\item
For each $s\in[0,1]$, $\Phi(s)$ is time reversal invariant.
\item  
 For each $k\in\bbN$ and $s\in[0,1]$, 
 the spectrum $\sigma\lmk \lmk H_{\Phi(s)+\Psi_k(s)}\rmk_{\Lambda_{n_k}}\rmk$
of $ \lmk H_{\Phi(s))+\Psi_k(s)}\rmk_{\Lambda_{n_k}}$ is decomposed into two non-empty disjoint parts
$
\sigma\lmk \lmk H_{\Phi(s)+\Psi_k(s)}\rmk_{\Lambda_{n_k}}\rmk=
\Sigma_1^{(k)}(s)\cup\Sigma_2^{(k)}(s)
$
such that $\Sigma_1^{(k)}(s)\subset [a(s),b(s)]$,
$\Sigma_2^{(k)}(s)\subset [b(s)+\gamma,\infty)$ and 
the diameter of $\Sigma_1^{(k)}(s)$ converges to $0$ as $k\to\infty$.
\end{enumerate}
\end{defn}
The interaction $\Psi_k(s)$ corresponds to a boundary condition. Note that it does not forbid an interaction
between intervals $[-n,-n+R]\cap \bbZ$ and $[n-R,R]\cap \bbZ$.
In particular, the periodic boundary condition is included in this framework.
Also, note that we {\it do not} require that the boundary term $\Psi_k(s)$
to be time reversal invariant. 
{\new{
The treatment including the boundary condition is rather new to this article
compared to other ones
\cite{bmns}\cite{nsy}.
}}

In section \ref{c1}, we will see the following. 
\begin{prop}\label{aep}
Let $\Phi:[0,1]\ni s \to \Phi(s):=\{\Phi(X;s)\}_{X\in {\mathfrak S}_\bbZ}\in {\caB_f}$ be a $C^1$-path of time reversal invariant gapped interactions satisfying the {\it Condition B}.
Let $\varphi_s$ be the unique $\tau^{\Phi(s)}$-ground state, for each $s\in[0,1]$. 
Then
$\varphi_0$ and $\varphi_1$
are automorphic equivalent via a time reversal invariant automorphism, 
which allows a time reversal invariant factorization.
\end{prop}
Note that {\it Condition B} implies
that for each $s\in[0,1]$, $\Phi(s)$ has a unique gapped ground state.
As a corollary of this proposition and Theorem \ref{aei}, we obtain the following.
\begin{thm}\label{c1t}
Let $\Phi:[0,1]\ni s \to \Phi(s):=\{\Phi(X;s)\}_{X\in {\mathfrak S}_\bbZ}\in {\caB_f}$ be a $C^1$-path of time reversal invariant gapped interactions satisfying the {\it Condition B}.
Then
we have
$\hat \sigma_{\Phi(0)}=\hat \sigma_{\Phi(1)}$.
\end{thm}
Namely, the $\bbZ_2$-index is invariant along the $C^1$-path of time reversal invariant gapped interactions, satisfying the {\it Condition B}.


\section{Proof of Theorem \ref{sp} and Theorem \ref{aei}}\label{split}

In order to introduce the $\bbZ_2$-index, let us note the following fact.
\begin{lem}\label{Wegner}
Let $\caH$ be a Hilbert space and $\Theta$ an antilinear  $*$-automorphism on $B(\caH)$
such that $\Theta^2=\id$.
Then there is an antiunitary operator $J$ on $\caH$ such that
\[
\Theta\lmk x\rmk=J xJ^*,\quad x\in B(\caH).
\]
Furthermore, $J^2=\sigma\unit $ with either $\sigma=1$ or $\sigma=-1$.
\end{lem}
\begin{proof}
First part is the Wigner's Theorem. 
For the second part, note that $J^2$ is a linear unitary operator which commute with
all the elements in $B(\caH)$.
Therefore, there is some $\sigma\in\bbT$ such that $J^2=\sigma\unit$.
We then have
\[
\sigma J=J^2\cdot J=J^3=J \cdot J^2=J\cdot \sigma\unit=\bar\sigma J.
\]
From this, we obtain $\sigma\in\{-1,1\}$.
\end{proof}
As a type I factor is isomorphic to $B(\caK)$ with some Hilbert space $\caK$, we may apply the Lemma \ref{Wegner} to obtain Theorem \ref{sp}.
\begin{proofof}[Theorem \ref{sp}]
From the time reversal invariance of $\varphi$, as in \cite{ot},
$\Xi_R$ has an extension $\hat \Xi_R$ to the von Neumann algebra $\pi_{\varphi_R}\lmk \caA_R\rmk{''}$,
as an antilinear  $*$-automorphism so that
\[
\hat\Xi_R\circ\pi_{\varphi_R}\lmk A\rmk=\pi_{\varphi_R}\circ\Xi_R\lmk A\rmk,\quad
A\in \caA_R,\quad\hat\Xi_R^2=\id.
\]
On the other hand, because $\pi_{\varphi_R}\lmk \caA_R\rmk{''}$ is a type I factor
 there exists a Hilbert space ${\caK_\varphi}$ and
a $*$-isomorphism $\iota_\varphi:\pi_{\varphi_R}\lmk \caA_R\rmk{''}\to B\lmk {\caK_\varphi}\rmk$
(\cite{takesaki}).
Then
\[
\Theta:=\iota_\varphi\circ \hat\Xi_R\circ\iota_\varphi^{-1}:
B\lmk {\caK_\varphi}\rmk\to B\lmk {\caK_\varphi}\rmk
\]
defines an antilinear  $*$-automorphism on $B({\caK_\varphi})$
such that $\Theta^2=\id$.
Applying Lemma \ref{Wegner},
we obtain an antiunitary $J_{\varphi}$ on $\caK_{\varphi}$ satisfying
$\Theta\lmk x\rmk=J_\varphi xJ_\varphi^*$, $x\in B(\caK_{\varphi})$
and $J_{\varphi}^2=\sigma_\varphi\unit$, with
$\sigma_\varphi\in\{-1,1\}$.
By the definition of $\Theta$ and $J_\varphi $, we obtain
\begin{align*}
\iota_\varphi\circ \pi_{\varphi_R}\circ\Xi_R\lmk A\rmk
=J_\varphi \lmk \iota_\varphi \circ\pi_{\varphi_R}\lmk A\rmk\rmk J_\varphi^*,\quad
A\in\caA_R.
\end{align*}
This proves the first half of Theorem \ref{sp}.
To prove the latter half of the statement,
suppose that $(\tilde \caK_\varphi, \tilde \iota_\varphi,\tilde  J_\varphi,\tilde\sigma_\varphi )$
satisfies the same conditions.
Then $\tilde \iota_\varphi\circ \iota_\varphi^{-1}:
B(\caK_\varphi)\to B(\tilde \caK_\varphi)$ is a linear
$*$-isomorphism. Therefore, by the Wigner's theorem, there exists a unitary 
$W:\caK_\varphi\to \tilde \caK_\varphi$
such that
$W xW^*=\tilde \iota_\varphi\circ \iota_\varphi^{-1}\lmk x\rmk$, $x\in B(\caK_{\varphi})$.
Using this $W$, 
\begin{align*}
&\tilde J_\varphi W\lmk \iota_\varphi \circ\pi_{\varphi_R}\lmk A\rmk\rmk W^*\tilde J_\varphi^*
=\tilde J_\varphi \lmk \tilde\iota_\varphi \circ\pi_{\varphi_R}\lmk A\rmk\rmk \tilde J_\varphi^*
=\tilde \iota_\varphi\circ \pi_{\varphi_R}\circ\Xi_R\lmk A\rmk\\
&=W\lmk
 \iota_\varphi\circ \pi_{\varphi_R}\circ\Xi_R\lmk A\rmk
 \rmk
W^*
=W\lmk J_\varphi \lmk \iota_\varphi \circ\pi_{\varphi_R}\lmk A\rmk\rmk J_\varphi^*
\rmk W^*,
\end{align*}
for all $A\in\caA_R$.
From this, we obtain
\[
\tilde J_\varphi W x W^*\tilde J_\varphi^*
=W J_\varphi  x J_\varphi^*
 W^*,\quad x\in B(\caK_\varphi).
\]
Hence, $J_\varphi^*W^*\tilde J_\varphi W: B(\caK_\varphi)\to B(\caK_\varphi)$ is a
linear unitary operator on $\caK_\varphi$ which commutes with any element of
$B(\caK_\varphi)$.
Therefore, there is $e^{i\theta}\in \bbT$ such that
$J_\varphi^*W^*\tilde J_\varphi W=e^{-i\theta}\unit$.
Hence we have $\tilde J_\varphi=e^{i\theta} WJ_\varphi W^*$ and
\begin{align}
\tilde J_\varphi^2=e^{i\theta} WJ_\varphi W^*e^{i\theta} WJ_\varphi W^*
=e^{i\theta-i\theta}WJ_\varphi^2W^*
=\sigma_\varphi\unit.
\end{align}
\end{proofof}
Hence we have defined the $\bbZ_2$-index for time reversal invariant pure states satisfying the pure split property.
Next we show that this $\bbZ_2$-index is an invariant of the 
automorphic equivalence via a time reversal invariant automorphism which allows time reversal invariant factorization, Theorem \ref{aei}.
\begin{proofof}[Thorem \ref{aei}]
Let $\varphi_1,\varphi_2$ be time reversal invariant pure states satisfying the split property.
Assume $\varphi_2$ is
 automorphic equivalent to $\varphi_1$ via a time reversal invariant automorphism $\alpha$, i.e.,
$\varphi_2=\varphi_1\circ\alpha$.
Assume that $\alpha$ allows a time reversal invariant factorization, i.e.,
there exist automorphisms $\alpha_R$,
$\alpha_L$ on $\caA_R$, $\caA_L$ and a unitary $W$ in $\caA$
such that
\begin{align}
&\alpha\circ\lmk\alpha_L^{-1}\otimes \alpha_R^{-1}\rmk(A)=
WAW^*,\quad A\in\caA,\label{aaw}\\
&\alpha_R\circ\Xi_R=\Xi_R\circ\alpha_R,\quad
\alpha_L\circ\Xi_L=\Xi_L\circ\alpha_L\label{aral}.
\end{align}

Let $\varphi_L$, 
$\varphi_R$, be the restriction of $\varphi_1$ to
$\caA_L$, $\caA_R$, respectively.We claim that $\varphi_2\vert_{\caA_R}\circ \alpha_R^{-1}$ and $\varphi_R$ are quasi-equivalent.
By (\ref{aaw}), the states
\[
\varphi_2\circ \lmk\alpha_L^{-1}\otimes \alpha_R^{-1}\rmk
=\varphi_1\circ\alpha\circ \lmk\alpha_L^{-1}\otimes \alpha_R^{-1}\rmk
=\varphi_1\circ\Ad W
\]
and $\varphi_1$ are quasi-equivalent.
As $\varphi_1$ satisfies the split property, by the proof of Proposition 2.2 of \cite{Matsui1},
$\varphi_L\otimes\varphi_R$ is quasi-equivalent to $\varphi_1$. (In Proposition 2.2 of \cite{Matsui1}, it is assumed that the state to be translationally invariant because of the first equivalent condition (i). However, the proof for the equivalence
(ii) and (iii) does not require translation invariance.)
Hence $\varphi_2\circ \lmk\alpha_L^{-1}\otimes \alpha_R^{-1}\rmk$ and $\varphi_L\otimes\varphi_R$
are quasi-equivalent.
Let $(\caH_L,\pi_L,\Omega_L)$, $(\caH_R,\pi_R,\Omega_R)$
be the GNS triple of $\varphi_L$, 
$\varphi_R$, respectively.
Note that $(\caH_L\otimes\caH_R,\pi_L\otimes\pi_R,\Omega_L\otimes\Omega_R)$
is the GNS triple of $\varphi_L\otimes\varphi_R$.
As $\varphi_2\circ \lmk\alpha_L^{-1}\otimes \alpha_R^{-1}\rmk$ and $\varphi_L\otimes\varphi_R$
are quasi-equivalent, there is a density matrix $\rho$ on $\caH_L\otimes\caH_R$
such that 
\begin{align}\label{rho}
\varphi_2\circ \lmk\alpha_L^{-1}\otimes \alpha_R^{-1}\rmk(A)=
\Tr_{\caH_L\otimes\caH_{R}}\lmk \rho\lmk \pi_L\otimes\pi_R\rmk (A)
\rmk,\quad A\in\caA.
\end{align}
Let
\begin{align}
\sigma:=\Tr_{\caH_L}\lmk\rho\rmk
\end{align}
be the reduced density matrix of $\rho$ on $\caH_R$.
Here $\Tr_{\caH_L}$ deontes the partial trace over $\caH_L$.
Substituting $A=\unit\otimes B$ with $B\in\caA_R$ in (\ref{rho}),
we obtain
\begin{align}
\varphi_2\vert_{\caA_R}\circ\alpha_R^{-1}(B)
=\Tr_{\caH_R}\lmk
\sigma\pi_R(B)
\rmk,\quad B\in\caA_R.
\end{align}
Hence, $\varphi_2\vert_{\caA_R}\circ\alpha_R^{-1}$ is $\varphi_R$-normal
.
As $\varphi_R$ is a factor state, $\varphi_2\vert_{\caA_R}\circ\alpha_R^{-1}$ and $\varphi_R$
and quasi-equivalent, proving the claim.

For $i=1,2$, let $(\caH_i,\pi_i,\Omega_i)$, be the GNS triple of $\left.\varphi_i\right\vert_{\caA_R}$,
and $\caK_{\varphi_i}$, $\iota_{\varphi_i}$, $\caJ_{\varphi_i}$, $\sigma_{\varphi_i}$ the objects given in
Theorem \ref{sp}.
Note that $(\caH_2,\pi_2\circ\alpha_R^{-1},\Omega_2)$ is the GNS triple of 
$\left.\varphi_2\right\vert_{\caA_R}\circ \alpha_R^{-1}$.
As $\varphi_R=\varphi_1\vert_{\caA_R}$ and $\left.\varphi_2\right\vert_{\caA_R}\circ \alpha_R^{-1}$ are quasi-equivalent,
by Theorem 2.4.26 of \cite{BR1}, there exists an $*$-isomorphism
$\tau: \pi_1\lmk \caA_R\rmk''\to \lmk \pi_2\circ\alpha_R^{-1}\lmk \caA_R\rmk\rmk''
= \lmk \pi_2\lmk\caA_R\rmk\rmk''$
such that
\begin{align}\label{tau}
\tau\lmk\pi_1\lmk A\rmk\rmk
=\pi_2\circ\alpha_R^{-1}\lmk A\rmk,\quad A\in\caA_R.
\end{align}
By Wigner's theorem, for the $*$-isomorphism
$\iota_{\varphi_2 }\circ\tau\circ \iota_{\varphi_1}^{-1}: B(\caK_{\varphi_1})\to B(\caK_{\varphi_2})$
there exists a unitary $U:\caK_{\varphi_1}\to \caK_{\varphi_2}$
such that 
\begin{align}\label{ut}
U xU^*=\iota_{\varphi_2 }\circ\tau\circ\iota_{\varphi_1}^{-1}\lmk x\rmk,\quad x\in B(\caK_{\varphi_1}).
\end{align}
By  (\ref{aral}), (\ref{tau}) and (\ref{ut}) and Theorem \ref{sp},
for any $A\in\caA_R$, we have
\begin{align*}
&J_{\varphi_2}U\lmk
\iota_{\varphi_1}\circ \pi_{1}\lmk A\rmk
\rmk
U^*J_{\varphi_2}^*
=
J_{\varphi_2}\lmk
\iota_{\varphi_2}\circ\tau\circ \pi_{1}\lmk A\rmk
\rmk
J_{\varphi_2}^*
=
J_{\varphi_2}\lmk
\iota_{\varphi_2}\circ \pi_{2}\circ\alpha_R^{-1}\lmk A\rmk
\rmk
J_{\varphi_2}^*
=
\iota_{\varphi_2}\circ \pi_{2}\circ\Xi_R\circ\alpha_R^{-1}\lmk A\rmk\\
&=
\iota_{\varphi_2}\circ \pi_{2}\circ\alpha_R^{-1}\circ\Xi_R\lmk A\rmk
= U
\lmk
\iota_{\varphi_1}\circ \pi_{1}\circ\Xi_R\lmk A\rmk
\rmk
U^*
= UJ_{\varphi_1}
\lmk
\iota_{\varphi_1}\circ \pi_{1}\lmk A\rmk
\rmk
J_{\varphi_1}^*
U^*.
\end{align*}
Multiplying $J_{\varphi_1}^*U^*$ from the left
and $J_{\varphi_2} U$ from the right of this equation,
we obtain
\begin{align*}
 \lmk J_{\varphi_1}^*U^*J_{\varphi_2} U\rmk \lmk \iota_{\varphi_1}\circ \pi_{1}\lmk A\rmk\rmk
=\lmk
\iota_{\varphi_1}\circ \pi_{1}\lmk A\rmk
\rmk \lmk J_{\varphi_1}^*U^*J_{\varphi_2} U\rmk,\quad
A\in\caA_R.
\end{align*}
Hence,  $J_{\varphi_1}^*U^*J_{\varphi_2} U$ is a untary operator on $\caK_{\varphi_1}$
which commutes with any bounded operator on $\caK_{\varphi_1}$.
Therefore, there exists $c\in\bbT$ such that $J_{\varphi_1}^*U^*J_{\varphi_2} U=c\unit$.
We then have
$cU^*J_{\varphi_2} U=J_{\varphi_1}$, and we obtain $\sigma_{\varphi_1}=\sigma_{\varphi_{2}}$.
\end{proofof}

\section{Proof of Proposition \ref{aep}}\label{c1}

In order to prove Proposition \ref{aep}, we use the tools provided in 
\cite{bmns}, which is based on Hastings's quasi-adiabatic continuation \cite{h1}.
Let $P_{k}(s)$ be
the spectral projection of $(H_{\Phi(s)+\Psi_k(s)})_{\Lambda_{n_k}}$
corresponding to the $\Sigma^{(k)}_{1}$ part in Definition \ref{boundary}.
From \cite{bmns} (Proposition 2.4 and Corollary 2.8),
there is a one parameter family of unitaries
$U_{k}(s)\in \caA_{\Lambda_{n_k}}$ such that $P_k(s)=U_k(s)P_k(0)U_k^*(s)$.
This $U_k$ is the solution of the differential equation
\begin{align}\label{udf}
-i\frac{d}{ds} U_{k}(s)=D_{k}(s) U_{k}(s),\quad U_k(0)=\unit.
\end{align}
Here, $D_{k}(s)$ is defined by
\begin{align}\label{dkdef}
D_{k}(s):=
\int_{-\infty}^\infty dt\;
W_\gamma(t)\;\; \tau_t^{\Phi(s)+\Psi_k(s),\Lambda_{n_k}} 
\lmk \frac{d}{ds}\lmk H_{\Phi(s)+\Psi_k(s)}\rmk_{\Lambda_{n_k}}\rmk,\quad s\in[0,1],
\end{align}
with $W_\gamma\in L^1(\bbR)$ being a real-valued odd function
such that $|W_\gamma(t)|$ is continuous, monotone decreasing for $t\ge 0$, and
\begin{align}\label{wg}
\int_0^\infty dt \int_t^\infty  \lv W_\gamma(t)\rv<\infty.
\end{align}
We set
\begin{align}\label{ig}
I_\gamma(t):=\int_t^\infty \lv W_\gamma(s)\rv ds,\quad t>0.
\end{align}
Similarly, we consider a one parameter family of unitaries 
$\hat U_{k,i}(s)$, $i=o,L,R$ 
which is the solution of
the differential equation
\begin{align}\label{hudf}
-i\frac{d}{ds} \hat U_{k,i}(s)=\hat D_{k,i}(s) \hat U_{k,i}(s),\quad \hat U_{k,i}(0)=\unit.
\end{align}
Here, $\hat D_{k,i}(s)$ is defined by
\begin{align}\label{hd}
\hat D_{k,i}(s):=
\int_{-\infty}^\infty dt\;
W_\gamma(t)\;\; \tau_t^{\Phi(s),I_{k,i}} 
\lmk \frac{d}{ds}\lmk H_{\Phi(s)}\rmk_{I_{k,i}}\rmk,\quad s\in[0,1],
\end{align}
with same $W_\gamma(t)$ as in (\ref{dkdef}).
In (\ref{hd}), we set  $I_{k,o}:=\Lambda_{n_k}$, $I_{k,L}:=\Lambda_{n_k}\cap(-\infty,-1]$,  and 
$I_{k,R}:=\Lambda_{n_k}\cap[0,\infty)$.
Let $\alpha_{s}^{(k)},\hat\alpha_s^{(k,i)}$ for $k\in\nan$, $i=o,L,R$, $s\in[0,1]$ 
 be automorphisms on $\caA$ given by
\begin{align*}
\alpha_s^{(k)} (A)=U_{k}(s)^* AU_k(s),\quad
\hat \alpha_s^{(k,i)} (A)=\hat U_{k,i}(s)^* A\hat U_{k,i}(s),\quad A\in\caA.
\end{align*}
From Definition \ref{boundary} {\it 6.}, the definitions of these automorphisms,
and the oddness of $W_\gamma(t)$,
$\hat \alpha_s^{(k,i)} $ commute with $\Xi$.
As we did not assume the time reversal invariance of $\Psi_k(s)$,
$\alpha_s^{(k)}$ does not need to commute with $\Xi$.
By \cite{bmns} proof of Theorem 5.2, for each $s\in[0,1]$,
 there exists the thermodynamic limits 
 $\alpha_{s,i}$
 of
$\hat \alpha_s^{(k,i)}$ for $i=o,L,R$:
\begin{align}\label{hal}
\lim_{k\to \infty}
\sup_{s\in[0,1]}\lV
\alpha_{s,i}(A)-\hat \alpha_s^{(k,i)}(A)
\rV=0,\quad
A\in\caA,\quad i=o,L,R,\quad s\in[0,1].
\end{align}
The limits $\alpha_{s,i}$ also commute with $\Xi$.

The automorphism
$\alpha_s^{(k)} $ also strongly converges to $\alpha_{s,o}$. Note that the difference between $\alpha_{s}^{(k)}$ and $\hat \alpha_s^{(k,o)}$
is just the boundary terms which goes to infinity far away as $k\to \infty$.
\begin{lem} \label{aaak}For any $A\in\caA$, we have
\begin{align}\label{bng}
\lim_{k\to \infty}
\lV
\alpha_{s,o}(A)-\alpha_s^{(k)}(A)
\rV=0.
\end{align}
\end{lem}
The proof of this Lemma is in Appendix \ref{lemmas}.
In the setting of Definition \ref{boundary}, 
let $\caS_{k}(s)$ be a set of states on $\caA_{\Lambda_{n_k}}$
whose support is under $P_k(s)$, $s\in[0,1]$, $k\in\nan$.
Because of the weak*-compactness of the state space, any sequence of extensions of
$\omega_{k,0}\in \caS_{k}(0)$ to $\caA$ has a weak*-accumulation point.
Due to the Definition \ref{boundary} {\it 7.},
any weak*-accumulation point of such sequence is
the $\tau^{\Phi(0)}$ ground state. From Definition \ref{boundary} {\it 5.}, it  is equal to $\varphi_0$.
As this holds for any weak*-accumulation point, we conclude that any extensions of $\omega_{k,0}$
converges to $\varphi_0$ in the weak*-topology.
By \cite{bmns} Corollary 2.8 , $\omega_{k,0}\circ \alpha_s^{(k)}$
is an element of $\caS_{k}(s)$, for each $s\in[0,1]$.
By the same reasoning as above, their extensions converges to $\varphi_s$ in the weak*-topology.
Using (\ref{bng}), as in \cite{bmns} Theorem 5.5, we conclude that $\varphi_s=\varphi_0\circ\alpha_{s,o}$.
Hence $\varphi_s$ is  automorphic equivalent to $\varphi_0$ via the time reversal invariant
automorphism $\alpha_{s,o}$.

Now let us prove that $\alpha_{s,o}$ is factorizable.
For an interaction $\Psi$, we introduce a new interaction $\tilde \Psi$ which is defined by
\begin{align}
\tilde \Psi(X):=\left\{
\begin{gathered}
\Psi(X),\quad \text{if}\quad  X\subset[0,\infty)\quad\text{or} \quad X\subset(-\infty,-1]\\
0,\quad \text{otherwise}
\end{gathered}
\right.
.
\end{align}
Namely, we remove the interaction between the left-infinite chain and the right infinite chain.
We set
\begin{align}
V_k(s):=&\hat D_{k,L}(s)+\hat D_{k,R}(s)-\hat D_{k,o}(s)\nonumber\\
=&
\int_{-\infty}^\infty dt\;
W_\gamma(t)\;\;
\lmk
\begin{gathered}
\sum_{X\subset \Lambda_{n_k}}
\lmk
 \tau_t^{\tilde \Phi(s),\Lambda_{n_k}} \lmk \tilde \Phi'(X;s)\rmk
- \tau_t^{\Phi(s),\Lambda_{n_k}} \lmk  \Phi'(X;s)\rmk
\rmk
\end{gathered}\rmk,
\end{align}
for $k\in\nan$ and 
$s\in[0,1]$. 
These $V_k(s)\in\caA_{\Lambda_{n_k}}$, $k\in\bbN$
converge to some self-adjoint opearator $V(s)$,
uniformly in $s\in[0,1]$, as $k\to\infty$.:
\begin{lem}\label{vvvk}
For each $s\in[0,1]$, there exists a self-adjoint element $V(s)\in\caA$
such that
\begin{align}\label{vcon}
\lim_{k\to\infty}\sup_{s\in [0,1]}
\lV
V_k(s)-V(s)
\rV=0.
\end{align}
\end{lem}
The proof of this Lemma is shown in Appendix \ref{lemmas}.
Furthermore, combining Lemma \ref{vvvk} with (\ref{hal}), 
we obtain
\begin{align}\label{avcon}
\lim_{k\to \infty}\sup_{s\in[0,1]}\lV\hat \alpha_s^{(k,o)}\lmk V_k(s)\rmk -
\alpha_{s,o}\lmk V(s)\rmk
\rV=0.
\end{align}
As a uniform limit of continuous functions, $V(s)$ and $\alpha_{s,o}\lmk V(s)\rmk$ are continuous in $s\in[0,1]$.

For each $k\in\nan$, let $W_k:[0,1]\to \caA_{\Lambda_{n_k}}$ be the solution of the differential equation
\begin{align}\label{wdrf}
\frac{dW_k(s)}{ds}
=i\hat \alpha_s^{(k,o)}\lmk V_k(s)\rmk W_k(s),\quad W_k(0)=\unit.
\end{align}
Then $W_{k}(s)$ is unitary and 
from (\ref{hudf}) and (\ref{wdrf}), 
we can check
\begin{align}
\hat \alpha_s^{(k,o)}\circ\lmk
\lmk \hat \alpha_{s}^{(k,L)}\rmk^{-1}\otimes\lmk \hat \alpha_s^{(k,R)}\rmk^{-1}
\rmk(A)
=W_k(s)AW_k(s)^*.
\end{align}
Because of the uniform convergence of $\hat \alpha_s^{(k,o)}\lmk V_k(s)\rmk$ from (\ref{avcon}), $W_k(s)$ also converges to 
a unitary $W(s)\in\caA$, uniformly in $s\in[0,1]$.
Combining this with the convergence of $\hat \alpha_s^{k,o}$,
$\alpha_{s}^{k,L}$,  $\alpha_s^{k,R}$, we obtain
\begin{align}
\alpha_{s,o}\circ\lmk
\lmk \alpha_{s,L}\rmk^{-1}\otimes\lmk \alpha_{s,R}\rmk^{-1}
\rmk(A)
=W(s)AW(s)^*,\quad A\in\caA,\quad s\in[0,1].
\end{align}
Hence $\alpha_{s,o}$ is factorizable with a time reversal invariant factorization
$(\alpha_{s,R},\alpha_{s,L},W(s))$, completing the proof of Proposition \ref{aep}.

\section{$\bbZ_2$-index for Matrix product states}\label{oshikawa}
In this section, we prove that the $\bbZ_2$-index $\sigma_\varphi$
for a matrix product state $\varphi$ is same as the $\bbZ_2$-index found in \cite{po}.
Throughout this section $S$ is an integer.(See \cite{ot}.)
First let us recall matrix product states.
Let $k\in\bbN$ be a number
and $\vv=(v_{\mu})_{\mu\in\SSS}\in \Mat_k^{\times (2S+1)}$ a $2S+1$-tuple of
$k\times k$ matrices.
For each $l\in\nan$, we set
\begin{equation}
{\mathcal K}_l(\vv) :=\spn\left\{v_{\mu_0}v_{\mu_{1}}\ldots v_{\mu_{l-1}}\mid
(\mu_0,\mu_1,\ldots,\mu_{l-1})\subset\SSS^{\times l}\right\}.
\end{equation}
We say $\vv$ is primitive if ${\mathcal K}_l(\vv)=\Mat_k$ for $l$ large enough.
We denote by $\Primz_u(2S+1,k)$ the set of 
all primitive $2S+1$-tuples $\vv$ of $k\times k$ matrices
such that 
\[
\sum_{\mu\in\SSS} v_\mu v_\mu^*=1.
\]

For $\vv\in\Primz_u(2S+1,k)$, there exists a unique $T_\vv$-invariant state
$\rho_\vv$. (See \cite{Wolf:2012aa} for example.)
Each $\vv\in \Primz_u(2S+1,k)$ generates a translationally invariant state $\omega_\vv$ by
\begin{align}
\omega_{\vv}\lmk
\bigotimes_{i=1}^{l}
e_{\mu_i,\nu_i}
\rmk=
\rho_\vv\lmk v_{\mu_1}\cdots v_{\mu_l} v_{\nu_l}^*\cdots v_{\nu_1}^*\rmk,\quad
\mu_i,\nu_i\in\SSS,\quad i=1,\ldots,l,\quad l\in\nan. 
\end{align}
A translationally invariant state which has this representation is called a matrix product state.
For a matrix product state, this representation is unique up to
unitary and phase:
If both of $\vv^{(1)}\in \Primz_u(2S+1,k_1)$ and $\vv^{(2)}\in \Primz_u(2S+1,k_2)$
generate the same matrix product state, then $k_1=k_2$
and there exist a unitary $U:\cc^{k_1}\to\cc^{k_2}$ and $e^{i\theta}\in\bbT$
such that
\begin{align}\label{unique}
U v_{\mu}^{(1)}=e^{i\theta} v_{\mu}^{(2)}U,\quad
\mu\in\SSS.
\end{align}

Let $\omega$ be a time reversal invariant matrix product state generated by $\vv\in \Primz_u(2S+1,k)$.
It is a unique ground state of some translation invariant finite range interaction.
That is, there is an interaction
$\Phi_\vv$ given by some fixed local positive element $h_\vv\in\caA_{[0,m-1]}$ with some
$m\in\nan$ as
\begin{align}\label{hamdef}
\Phi_{\vv}(X):=\left\{
\begin{gathered}
\beta_x\lmk h_\vv\rmk,\quad \text{if}\quad  X=[x,x+m-1]\cap\bbZ \quad \text{for some}\quad  x\in\bbZ\\
0,\quad\text{otherwise}
\end{gathered}\right.
\end{align}
for each $X\in {\mathfrak S}_{\bbZ}$ and $\omega$ is a unique $\tau^{\Phi_\vv}$-ground state. (See \cite{Fannes:1992vq} and \cite{Ogata3}.)
For this interaction $h_\vv$, $1-h_{\vv}$ is equal to 
the support of $\omega\vert_{\caA_{[0,m-1]}}$. (See  the proof of Lemma 3.19 of \cite{Ogata1} equation (48). Note that primitive $\vv$ belongs to $\ClassA$, Remark 1.16 of \cite{Ogata1}).
Therefore, from the time reversal invariance of $\omega$, $h_\vv$ satisfies 
\begin{align}\label{hinv}
\Xi(h_\vv)=h_\vv.
\end{align}
The Hamiltonian given by this interaction is frustration-free,
i.e.,
for each finite interval $I$ with $|I|\ge m$,
the local Hamiltonian $\lmk H_{\Phi_\vv}\rmk_{I}$ has a nontrivial kernel, which is the ground state space
of $\lmk H_{\Phi_\vv}\rmk_{I}$.
We denote by $G_{I,\vv}$, the orthogonal projection onto this kernel.
By Lemma 3.19 of \cite{Ogata1}, and its proof (equation (48)),
the support of the restriction $\left.\omega\right\vert_{\caA_I}$
is equal to $G_{I,\vv}$ and  there exists some constant $d_\vv>0$ such that
\begin{align}\label{fb}
\psi\le d_\vv \cdot \omega,
\end{align}
for any frustration free state $\psi$ on $\caA_R$, i.e., a state $\psi$ satisfying $\psi(\beta_x(h_\vv))=0$ for any
$0\le x\in\bbZ$.

%
%
%
%

We represent the statement of \cite{po}, in the way formulated by Tasaki \cite{TasakiBook}.
Let $\omega$ be a time reversal invariant matrix product state.
Let $\vv\in  \Primz_u(2S+1,k)$ be a generator of $\omega$, and $c$
a complex conjugation on $\bbC^k$ (i.e., an arbitrary anti-unitary with $c^2=\unit$).
By the time reversal invariance, one can see that $\tilde v_\mu:=(-1)^{S+\mu}cv_{-\mu} c$, $\mu\in\SSS$
also generates $\omega$. Note that state $ \rho_{\tilde \vv}(A):=\rho_\vv(cA^*c)$, $A\in\Mat_k$
is the $T_{\tilde\vv}$-invariant state.

From the uniqueness (\ref{unique}), 
there is a unitary $U$ on $\bbC^k$ and $e^{i\theta}\in\bbT$ such that
\begin{align}\label{cvu}
(-1)^{S+\mu}cv_{-\mu} c=e^{i\theta} U v_\mu U^*,\quad \mu\in\SSS.
\end{align}
In \cite{po}, it is shown that
\begin{align}\label{zetadef}
cUcU=\zeta_\omega\unit,\quad \text{with some} \quad \zeta_\omega\in\{-1,1\},
\end{align}
using the primitivity of $\vv$ and $S\in\nan$. (See \cite{TasakiBook}.)

We claim that this $\zeta_\omega$ does not depend on the choice of
$(\vv, c, U, e^{i\theta})$. To see this, suppose that
$(\vv_j, c_j, U_j, e^{i\theta_j})$, $j=1,2$ satisfy the above conditions.
By the uniqueness, there is a unitary $W$ and $e^{ir}\in \bbT$ such that
$Wv_{1\mu}=e^{ir}v_{2\mu}W$, $\mu\in\SSS$.
From this and (\ref{cvu}) (for $(\vv_j, c_j, U_j, e^{i\theta_j})$, $j=1,2$),
we have
\begin{align}
e^{i\theta_2-ir}U_2Wv_{1\mu}W^* U_2^*
=
(-1)^{S+\mu}c_2v_{2,-\mu}c_2
=e^{ir+i\theta_1} c_2Wc_1U_1v_{1\mu}U_1^*c_1W^*c_2.
\end{align}
Hence we obtain
\begin{align}\label{vvv}
e^{2ir+i\theta_1-i\theta_2}V v_{1\mu}V^*=v_{1\mu},\quad\mu\in\SSS,
\end{align}
with unitary $V:=W^*U_2^*c_2Wc_1U_1$.
From the primitivity of $\vv_{1}$, there are coefficients $c_{\mu_1,\ldots,\mu_l}\in\bbC$,
$\mu_i\in\SSS$, $i=1,\ldots,l$,
such that $\sum_{\mu_1,\ldots,\mu_l\in\SSS} c_{\mu_1,\ldots,\mu_l} v_{1\mu_1}\cdots v_{1\mu_l}=1$, for $l$ large enough.
From this and (\ref{vvv}), we see that $e^{2ir+i\theta_1-i\theta_2}=1$.
Substituting this to (\ref{vvv}), and from the primitivity of $\vv_1$, we obtain 
$V=e^{i\eta}\unit$, with scalar $e^{i\eta}\in \bbT$.
By the definition of $V$, we obtain $Wc_1U_1W^*=e^{-i\eta}c_2U_2$.
From this, we obtain
$Wc_1U_1c_1U_1W^*=c_2U_2c_2U_2$,
proving the claim.

This $\zeta_\omega$ is the $\bbZ_2$-index of \cite{po}.
As a matrix product state $\omega$
is pure and a unique gapped ground state by \cite{Fannes:1992vq},
 it satisfies the split property.
Therefore, we can associate $\omega$, our $\bbZ_2$-index $\sigma_\omega$ in Definition \ref{indexsp}.
We then have the following theorem.
\begin{thm}\label{poo}
For a time reversal invariant matrix product state $\omega$, we have
\[
\sigma_\omega=\zeta_\omega.
\]
\end{thm}
\begin{rem}
 {\new{It was shown in \cite{po} that 
that
$\zeta_{\varphi_{\rm AKLT}}=-1$ and $\zeta_{\varphi_{\rm trivial}}=1$.
From Theorem \ref{poo}, we conclude 
$\sigma_{\varphi_{\rm AKLT}}=-1$ and $\sigma_{\varphi_{\rm trivial}}=1$.
This is what we stated in the introduction.
}} 
\end{rem}

\begin{proof}Let $\omega$ be a time reversal invariant matrix product state generated by $\vv\in \Primz_u(2S+1,k)$.
Let $\omega_R$ be the restriction of $\omega$ to $\caA_R$,
 and $(\caH,\pi,\Omega)$ 
its GNS triple.
As $\omega$ is pure and split, 
$\pi(\caA_R)''$ is a type I factor.
Therefore, by Chapter V Theorem 1.31 of \cite{takesaki},
there are
separable Hilbert spaces $\caH_1,\caH_2$,
a representation $\pi_1$ of $\caA_R$ on $\caH_1$,
a unitary $W:\caH\to \caH_1\otimes \caH_2$
such that
\begin{align}
\hat\pi_1(A):=\pi_1(A)\otimes \unit=W\pi(A)W^*,\quad A\in \caA_R,
\end{align}
and $\pi_1(\caA_R)''=B(\caH_1)$.
Note that $(\caH_1\otimes \caH_2,\hat \pi_1,W\Omega)$ is a GNS representation of
$\omega_R$.
We denote by $\rho$, the reduced density matrix of $\ket{W\Omega}\bra{W\Omega}$,i.e.,
\begin{align}\label{rhod}
\Tr_{\caH_2}\lmk
\ket{W\Omega}\bra{W\Omega}
\rmk=
\rho.
\end{align}
Here $\Tr_{\caH_2}$ denotes the partial trace over $\caH_2$.
As in the proof of Theorem \ref{sp}, there exists an anti-unitary $K_1$ on $\caH_1$
such that
\begin{align}\label{kpk}
\pi_1\lmk
\Xi_R\lmk
A
\rmk
\rmk
=K_1\pi_1(A)K_1^*,\quad A\in\caA_R.
\end{align}
For this $K_1$, we have $K_1^2=\sigma_\omega\unit$, by Theorem \ref{sp}.

As $\omega$ is translation invariant,
there exist operators $\cs_\mu\in \pi_1(\caA_R)''=B(\caH_1)$ with $\mu\in\SSS$
satisfying the following:
\begin{align}
&\cs_{\mu}^*\cs_{\nu}=\delta_{\mu\nu}\unit,\label{s1}\\
&\sum_{\mu\in\SSS} \cs_{\mu} \pi_{1}(A) \cs_{\mu}^*=\pi_1\circ\beta_{1}(A),\quad A\in \caA_R \label{st}.\\
&\pi_1\lmk e_{\mu\nu}\otimes\unit_{[1,\infty)}\rmk
=\cs_\mu \cs_\nu^*\quad
\text{for all} \quad \mu,\nu\in\SSS.\label{sr}
\end{align}
(See \cite{arv,bjp,BJ}, Proof of Proposition 3.5 of \cite{Matsui1} and Lemma 3.5 of \cite{Matsui3}.)
Here $ e_{\mu\nu}\otimes\unit_{[1,\infty)}$ indicates an element $e_{\mu\nu}$ in $\caA_{\{0\}}=\Mat_{2S+1}$
embedded into $\caA_R$.
From (\ref{st}) and (\ref{sr}), we have
\begin{align}\label{ssss}
\pi_1\lmk\bigotimes_{k=0}^{l-1}e_{\mu_k,\nu_k}\rmk
=\cs_{\mu_0}\cdots \cs_{\mu_{l-1}}\cs_{\nu_{l-1}}^*\cdots \cs_{\nu_0}^*,
\end{align}for all $l\in\nan$, $\mu_k,\nu_k\in\SSS$.

By the same argument as in the proof of Theorem 2.2 of \cite{ot},
we see that there is some $e^{i\theta}\in\bbT$ such that
\begin{align}\label{stran}
\cs_\mu= e^{-i\theta}(-1)^{S+\mu}K_1 \cs_{-\mu}K_1^*,\quad \mu\in\SSS.
\end{align}

Now we restrict these $\cs_\mu$ to {\it a frustration-free subspace}
$\caK$ of $\caH_1$.
Recall that $\omega$ is the frustration free ground state of the translation invariant finite range interaction $\Phi_\vv$ (\ref{hamdef}). Namely, there is a self-adjoint element $h_\vv\in \caA_{[0,m-1]}$, 
such that
$\omega(\beta_x(h_\vv))=0$ for all $x\in\bbZ$.
We consider the following frustration-free
subspace of $\caH_1$:
\begin{align*}
\caK:=\cap_{\bbZ\ni x\ge 0}\ker \pi_1\lmk
\beta_x\lmk h_\vv\rmk
\rmk.
\end{align*}
Note that the support of $\rho$, (\ref{rhod}),  is in $\caK$, because $\omega$ is frustration-free.
Let $P_\caK$ be the orthogonal projection onto $\caK$.
As in \cite{Matsui3} (Lemma 3.2 and the argument in the proof of Lemma 3.6), $\caK$ is a finite dimensional space,
and 
$\cs_\mu^*$ preserves $\caK$:
\begin{align}\label{psb}
\cs_\mu^* P_\caK=P_{\caK}\cs_\mu^* P_\caK ,\quad \mu\in\SSS.
\end{align}
We denote $(\cs_\mu^* P_\caK)^*$ by $B_\mu$, $\mu\in\SSS$.

We claim that $\bbB=(B_\mu)_{\mu\in\SSS}$ is primitive. To prove this, it suffices to show that
$\rho$ is faithful on $\caK$ and for the completely positive unital map $T_\bbB$ defined
by $T_\bbB(x)=\sum_{\mu\in\SSS} B_\mu x B_\mu^*$, $x\in B(\caK)$,
we have $T_\bbB^N(x)\to \rho(x)\unit $, as $N\to\infty$,
for each $x\in B(\caK)$. (See Lemma C.5 of \cite{Ogata1}.)
First we show that $\rho$ is faithful on $\caK$.
If $\rho$ is not faithful on $\caK$, 
then there exists a unit vector $\xi\in\caK$ which is orthogonal to
the support of $\rho$.
By the definition of $\caK$, this $\xi$ defines a frustration free state
$\psi=\braket{\xi}{\pi_1\lmk\cdot\rmk\xi}$.
Let $p$ be the orthogonal projection onto the one-dimensional space $\bbC\xi$.
As $\pi_1(\caA_R)''=B(\caH_1)$, by Kaplansky's density Theorem, there exists
a net $\{x_\alpha\}_{\alpha}$ of positive elements in the unit ball of $\caA_{R}$
such that $\pi_1\lmk x_\alpha\rmk\to p$ in the $\sigma w$-topology.
For this net, we have $\lim_{\alpha}\omega(x_\alpha)=0$
and $\lim_\alpha\psi(x_\alpha)=1$.
This contradicts to (\ref{fb}).
Hence $\rho$ is faithful on $\caK$.
Next we show $T_\bbB^N(x)\to \rho(x)\unit $, as $N\to\infty$ for all $x\in B(\caK)$.
By $\pi_1(\caA_R)''=B(\caH_1)$ and the finite dimensionality of $\caK$,
we have $B(\caK)=P_\caK \pi_1\lmk \caA_{R}\cap\caA_{\rm loc}\rmk P_\caK$.
Therefore, for each $x\in B(\caK)$, there is an element $A\in \caA_{R}\cap\caA_{\rm loc}$
such that $x=P_{\caK}\pi_1\lmk A \rmk P_\caK$.
As $\omega$ is a factor state and translation invariant, we have
$\sigma w-\lim_{N\to\infty}\pi_1\circ \beta_N(A)=\omega(A)\unit$.
Therefore, for any $\eta\in\caK$, we have
\begin{align}
\braket{\eta}{T_{\bbB}^N\lmk x\rmk\eta}
=
\braket{\eta}{T_{\bbB}^N\lmk P_{\caK}\pi_1 \lmk A\rmk P_{\caK}\rmk\eta}
=\braket{\eta}{\pi_1\circ\beta_N\lmk A\rmk\eta}
\to \omega(A)\lV \eta\rV^2
=\rho(x)\lV \eta\rV^2,\quad N\to\infty.
\end{align}
Hence $\bbB$ is primitive.

The above proof for the primitivity also tells us that $\rho$ is the $T_\bbB$-invariant state.
From (\ref{ssss}) and the definition of $\bbB$ and (\ref{psb}), we see that 
$\bbB$ is a $2S+1$-tuple generating $\omega$.

By (\ref{kpk}), (\ref{hinv}) and $\Xi\circ \beta_x=\beta_x\circ \Xi$, we obtain
\begin{align}
\pi_1\lmk \beta_x(h_\vv)\rmk K_1^*P_{\caK}
=K_1^*K_1 \pi_1\lmk \beta_x(h_\vv)\rmk K_1^*P_{\caK}
=K_1^* \pi_1\lmk\Xi_R\circ \beta_x(h_\vv)\rmk P_{\caK}
=K_1^* \pi_1\lmk \beta_x(h_\vv)\rmk P_{\caK}=0,
\end{align}
for any $0\le x\in\bbZ$.
From this, we obtain
\begin{align}\label{pkpkp}
P_\caK  K_1P_\caK=P_\caK K_1 .
\end{align}
Similarly, from 
\begin{align}
\pi_1\lmk \beta_x(h_\vv)\rmk K_1P_{\caK}
=K_1K_1^* \pi_1\lmk \beta_x(h_\vv)\rmk K_1P_{\caK}
= K_1\pi_1\lmk\Xi_R\circ \beta_x(h_\vv)\rmk P_{\caK}
=K_1 \pi_1\lmk \beta_x(h_\vv)\rmk P_{\caK}=0,
\end{align}
for any $0\le x\in\bbZ$,
we obtain 
\begin{align}\label{pkppk}
P_\caK  K_1P_\caK=K_1P_\caK .
\end{align}
Hence $P_\caK$ and $K_1$ commute.
Because of this, we may define an anti-unitary $K_2:=P_\caK  K_1P_\caK=P_\caK K_1=K_1 P_\caK$
on $\caK$.

Multiplying $P_\caK$ from left of (\ref{stran}),
and using (\ref{psb}), (\ref{pkpkp}), (\ref{pkppk}) and the definition of $\bbB$,
we obtain
\begin{align}\label{bkb}
B_\mu= e^{-i\theta}(-1)^{S+\mu}K_2 B_{-\mu}K_2^*,\quad \mu\in\SSS.
\end{align}
Choose some complex conjugation $c$ on $\caK$ and
define $U:=cK_2$. Then $U$ is an unitary on $\caK$ and multiplying
$c$ from left and right of
(\ref{bkb}), we obtain
\begin{align*}
(-1)^{S+\mu}cB_{-\mu}c=e^{i\theta} UB_\mu U^*,\quad
\mu\in\SSS.
\end{align*}
Namely, ($\bbB$, $c$, $U$, $e^{i\theta}$) satisfies the 
condition of the quadrapret to define the $\zeta_\omega$ (\ref{zetadef}).
Therefore, we have
$cUcU=\zeta_\omega P_\caK$.
We then get
\begin{align*}
\sigma_\omega P_\caK=K_1^2P_\caK=K_2^2=cUcU=\zeta_\omega P_\caK.
\end{align*}
Hence we obtain $\zeta_\omega=\sigma_\omega$

\end{proof}

{\bf Acknowledgment.}\\
{The author is grateful to Hal Tasaki for fruitful discussion which was essential for the present work.
This work was supported by JSPS KAKENHI Grant Number 16K05171. 
Part of this paper was written during the visit of the author to
 CRM, 
 with the support of CRM-Simons program “Mathematical challenges in many-body physics 
 and quantum information”.
}
\bigskip

\appendix
\section{Proof of Lemma \ref{aaak} and Lemma \ref{vvvk}}\label{lemmas}
In this section we prove Lemma \ref{aaak} and Lemma \ref{vvvk}.
The proof is based on arguments and tools in \cite{bmns}.
For $M$ in {\it Condition B}, we may and will assume that $M>2$ .
Let us first recall the Lieb-Robinson bound.Fix some $a>0$ (throughout this appendix), and
define a positive function $F_a(r)$ on $\bbR_{\ge 0}$ by $F_a(r):=(1+r)^{-2}e^{-ar}$.
For a path of interactions satisfying Definition \ref{boundary},there exist
positive constants
$C_{1,a}$, $v_a$ satisfying the following.:
For any $X,Y\in{\mathfrak S}_{\bbZ}$, $A\in\caA_X$,
$B\in\caA_Y$, $k\in\nan$, $s\in[0,1]$ and $t\in\bbR$, we have
\begin{align}\label{lr}
&\lV
\left[
\tau_{t}^{\Phi(s)}(A),B
\right]
\rV,
\lV
\left[
\tau_{t}^{\tilde \Phi(s)}(A),B
\right]
\rV,
\lV
\left[
\tau_{t}^{\Phi(s),\Lambda_{n_k}}(A),B
\right]
\rV,
\lV
\left[
\tau_{t}^{\tilde \Phi(s),\Lambda_{n_k}}(A),B
\right]
\rV,
\lV
\left[
\tau_{t}^{\Phi(s)+\Psi_{k}(s),\Lambda_{n_k}}(A),B
\right]
\rV
\nonumber\\
&\le
C_{1,a}e^{v_a|t|}
\sum_{x\in X,y\in Y}F_a(|x-y|)\lV A\rV \lV B\rV.
\end{align}
{\new{(The inequality means that each of the left hand side can be bounded by the same
value written on the right hand side. We use this way of writing below as well.)}}
As in the proof of Theorem 2.2 \cite{nos}, 
perturbation of dynamics can be estimated by the use of the Lieb-Robinson bound.
In particular, by the Lieb-Robinson bound (\ref{lr})
and {\it 2.} of Defintion \ref{boundary},
for the fixed $a>0$ above, there exists a constant $C_{2,a}$ such that
\begin{align}\label{tpl}
&\lV
\tau_t^{\tilde\Phi(s),\Lambda_{n}}\lmk A\rmk
-\tau_t^{\Phi(s),\Lambda_{n}}\lmk A\rmk
\rV
=\lV
\int_0^t du
\frac{d}{du}\lmk
\tau_{t-u}^{\tilde\Phi(s),\Lambda_{n}}\circ \tau_u^{\Phi(s),\Lambda_{n}}\lmk A\rmk
\rmk
\rV\nonumber\\
&=\lV
\int_0^t du
 \;\tau_{t-u}^{\tilde\Phi(s),\Lambda_{n}}\lmk
\sum_{\substack{X\cap[0,\infty)\neq\emptyset,\; 
X\cap (-\infty,-1]\neq\emptyset}}
i\left[
\Phi(X;s),\tau_u^{\Phi(s),\Lambda_{n}}\lmk A\rmk
\right]
\rmk
\rV\nonumber\\
&\le
C_{2,a}\sum_{y\in Y}e^{v_a|t|-a|y|}\lV A\rV
,
\end{align}
for all $t\in\bbR$, $s\in[0,1]$, $n\in\nan$,
$Y\in{\mathfrak S}_{\bbZ}$, and $A\in\caA_Y$.
Here, $v_a$ is the same constant as in (\ref{lr}).
Similarly, we have
\begin{align}\label{ppp}
&\lV
\tau_t^{\Phi(s),\Lambda_{n_k}}\lmk A\rmk
-\tau_t^{\Phi(s)+\Psi_k(s),\Lambda_{n_k}}\lmk A\rmk
\rV
\le
C_{3,a}\sum_{y\in Y}e^{v_a|t|-a\cdot d(y,\lmk \Lambda_{n_k-R}\rmk^c)}\lV A\rV
,
\end{align}
for all $t\in\bbR$, $s\in[0,1]$, $k\in\nan$,
$Y\in{\mathfrak S}_{\bbZ}$, and $A\in\caA_Y$.

Taking $n\to\infty$ limit in (\ref{tpl}), we obtain
\begin{align}\label{tppd}
\lV
\tau_t^{\tilde\Phi(s)}\lmk A\rmk
-\tau_t^{\Phi(s)}\lmk A\rmk
\rV
\le C_{2,a}\sum_{y\in Y}e^{v_a|t|-a|y|}\lV A\rV,
\end{align}
for all $t\in\bbR$, $s\in[0,1]$, 
$Y\in{\mathfrak S}_{\bbZ}$, and $A\in\caA_Y$.
This estimate tells us that if $A$ is far away from the origin of 
$\bbZ$ compared to $|t|$, the difference between
the dynamics given by $\Phi(s)$ and $\tilde\Phi(s)$ is small.

By the same argument as in (\ref{tpl}), for the fixed $a>0$, there exists
a positive constant $C_{3,a}$ such
\begin{align}\label{nos}
&\lV
\tau_t^{\tilde\Phi(s),\Lambda_{m}}\lmk A\rmk
-\tau_t^{\tilde \Phi(s),\Lambda_{n}}\lmk A\rmk
\rV,\;
\lV
\tau_t^{\Phi(s),\Lambda_{m}}\lmk A\rmk
-\tau_t^{\Phi(s),\Lambda_{n}}\lmk A\rmk
\rV,\nonumber\\
&\lV
\tau_t^{\tilde\Phi(s),\Lambda_{m}}\lmk A\rmk
-\tau_t^{\tilde \Phi(s)}\lmk A\rmk
\rV,\;
\lV
\tau_t^{\Phi(s),\Lambda_{m}}\lmk A\rmk
-\tau_t^{\Phi(s)}\lmk A\rmk
\rV
\nonumber\\
&\le
C_{3,a}e^{v_a|t|}
\sum_{y\in Y}\sum_{x\in\Lambda_{m}^c}
F_a\lmk
|x-y|
\rmk\lV A\rV
\end{align}
for all $n,m\in\nan$, $n>m$,
$t\in\bbR$, $s\in[0,1]$, 
$Y\in{\mathfrak S}_{\bbZ}$, and $A\in\caA_Y$.
Here, $v_a$ is the same constant as in (\ref{lr}).
For each $k\in\nan$, we denote by $m_k$ the
the smallest integer less than or equal to $n_k/2$.

\begin{proofof}[Lemma \ref{aaak}]
We first show that
\begin{align}\label{ibng}
\lim_{k\to \infty}
\lV
\lmk \hat \alpha_{s}^{(k,o)}\rmk^{-1}(A)-\lmk \alpha_s^{(k)}\rmk^{-1}(A)
\rV=0,
\end{align}
for any $l\in\nan$ and $A\in\caA_{\Lambda_l}$. 
Fix any  $l\in\nan$ and $A\in\caA_{\Lambda_l}$. 
We may and we will assume that $n_k\ge 4(M+R+l)$ for each $k\in\nan$.
For each $k\in\nan$, we have
\begin{align}
\frac{d}{ds}\alpha_s^{(k)}\circ\lmk \hat \alpha_s^{(k,o)}\rmk^{-1}\lmk A\rmk
=\alpha_s^{(k)}
\lmk
i\left[
-D_k(s)+\hat D_{k,o}(s),\;
\lmk \hat \alpha_s^{(k,o)}\rmk^{-1}\lmk A\rmk
\right]
\rmk.
\end{align}
We claim
\begin{align}\label{claim}
\varepsilon_k(A):=\sup_{s\in[0,1]}\lV\left[
-D_k(s)+\hat D_{k,o}(s),\;
\lmk \hat \alpha_s^{(k,o)}\rmk^{-1}\lmk A\rmk
\right]\rV\to0,\quad k\to\infty.
\end{align}
To show this, we split $\lmk \hat \alpha_s^{(k,o)}\rmk^{-1}\lmk A\rmk$ into two parts.
For each $k$, we denote by $L_k$, the smallest integer less than or equal to
$\frac{n_k}4$.
Recall also that $m_k$ is the smallest integer less than or equal to $\frac{n_k}2$.
From \cite{bmns} proof of Theorem 4.5 and Lemma 3.2, $\lmk \hat \alpha_s^{(k,o)}\rmk^{-1}\lmk A\rmk$ can be decomposed into
an element $\Pi_{L_k}\lmk \lmk \hat \alpha_s^{(k,o)}\rmk^{-1}\lmk A\rmk\rmk$
in $\caA_{L_k}$ 
with $\lV
\Pi_{L_k}\lmk \lmk \hat \alpha_s^{(k,o)}\rmk^{-1}\lmk A\rmk\rmk
\rV\le\lV A\rV$,
and the rest, which is bounded from above as
\begin{align}\label{dfr}
\lV
\lmk \hat \alpha_s^{(k,o)}\rmk^{-1}\lmk A\rmk-
\Pi_{L_k}\lmk \lmk \hat \alpha_s^{(k,o)}\rmk^{-1}\lmk A\rmk\rmk
\rV
\le
C_1\;(2l+1)\;\tilde u\lmk d\lmk
\Lambda_l,\Lambda_{n_k}\setminus \Lambda_{L_k}
\rmk\rmk\lV A\rV.
\end{align}
The function $\tilde u(r)$, $r>0$ on the right hand side satisfies $\tilde u(r)\to 0$,
as $r\to\infty$.

The difference $-D_k(s)+\hat D_{k,o}(s)$ is localized at the boundary of $\Lambda_{n_k}$. Therefore,
by Lieb-Robinson bound, it almost commutes with $\Pi_{L_k}\lmk \lmk \hat \alpha_s^{(k,o)}\rmk^{-1}\lmk A\rmk\rmk$ for $k$ large enough.
For simplicity, let us introduce a notation
\begin{align}
B(X,s,t,k):=
\tau_{t}^{\Phi(s),\Lambda_{n_k}} \lmk \Phi'(X;s)\rmk
-\tau_{t}^{\Phi(s)+\Psi_{k}(s),\Lambda_{n_k}} \lmk \Phi'(X;s)\rmk,
\end{align}
for $X\in{\mathfrak S}_{\bbZ}$, $t\in\bbR$, $s\in [0,1]$, and $k\in\bbN$.
We have
\begin{align}\label{bn}
-D_k(s)+\hat D_{k,o}(s)
=\sum_{X\subset \Lambda_{n_k}}
\int_{-\infty}^\infty dt \; W_\gamma(t)
B(X,s,t,k)
-\sum_{\substack{X\subset \Lambda_{n_k}\\X \subset\Lambda_{n_k}\setminus \Lambda_{n_k-R}}}
\int_{-\infty}^\infty dt \; W_\gamma(t)
\tau_{t}^{\Phi(s)+\Psi_{k}(s),\Lambda_{n_k}} \lmk \Psi_k'(X;s)\rmk.
\end{align}
Set
\begin{align}
T_X^k:=\frac a{2v_a} \cdot d\lmk X,\lmk \Lambda_{n_k-R}\rmk^c\rmk,\quad
S_X^k:=\frac a{2v_a} \cdot d\lmk X,\Lambda_{L_k}\rmk
\end{align}
for each $k\in\nan$ and $X\in{\mathfrak S}_\bbZ$.
We split the summation of $X\subset\Lambda_{n_k}$ in the first term of (\ref{bn}) into $X\subset \Lambda_{m_k}$ and
$X\cap\lmk \Lambda_{m_k}\rmk^c\neq\emptyset$.
For  $X\subset \Lambda_{m_k}$, we split the integration into
$|t|\le T_X^{k}$ part and $|t|\ge T_X^{k}$ part.
For $X\cap\lmk \Lambda_{m_k}\rmk^c\neq\emptyset$, we split the integration into
$|t|\le S_X^{k}$ part and $|t|\ge S_X^{k}$ part.

First we consider $X\subset \Lambda_{m_k}$ and
$|t|\le T_X^{k}$ part.
From (\ref{ppp}), and Definition \ref{boundary} {\it 2.}, we have 
\begin{align}\label{two}
&\lV
\sum_{X\subset \Lambda_{m_k}}
\int_{|t|\le T_X^k} dt W_\gamma(t) B(X,s,t,k)
\rV
\le
\sum_{\substack{X\subset \Lambda_{m_k}\\ \diam X<M}}\lV W_\gamma\rV_1
C_1C_{3,a}\sum_{y\in X}e^{v_aT_X^k-a\cdot d\lmk y,\lmk \Lambda_{n_k-R}\rmk^c\rmk}\nonumber\\
&\le
\sum_{\substack{X\subset \Lambda_{m_k}\\ \diam X<M}}\lV W_\gamma\rV_1
C_1C_{3,a}Me^{-\frac a2\cdot d\lmk X,\lmk \Lambda_{n_k-R}\rmk^c\rmk}
=C_1C_{3,a}M\lV W_\gamma\rV_1\sum_{j=n_k-m_k-R}^\infty\sum_{\substack{X\subset \Lambda_{m_k}\\ \diam X<M\\d\lmk X,\lmk \Lambda_{n_k-R}\rmk^c\rmk=j}}e^{-\frac a2 j}\nonumber\\
&\le
C_1C_{3,a}M2^{M}\lV W_\gamma\rV_1\sum_{j=n_k-m_k-R}^\infty 
e^{-\frac a2 j}
\end{align}
Note that for $X\subset \Lambda_{m_k}$, the distance between $X$ and
$\lmk \Lambda_{n_k-R}\rmk^c$ is at least $n_k-R-m_k$.
This is used in the equality in the second line.
Recall that $n_k-m_k-R\ge 1$ as we assumed $n_k\ge 4(M+R+l)$ in the beginning of the proof.
In the last inequality, we used the fact 
that for any $j\ge 1$,
the number of $X\subset\Lambda_{m_k}$ with $\diam(X)<M$
such that $d\lmk X,\lmk \Lambda_{n_k-R}\rmk^c\rmk=j$
is at most $2^{M}$.
Note that the last line of (\ref{two}) is independent of $s\in[0,1]$ and goes to $0$
as $k\to\infty$.

Next we estimate the first term of (\ref{bn}) corresponding to 
$X\cap\lmk \Lambda_{m_k}\rmk^c\neq\emptyset$ and
$|t|\le S_X^{k}$ part. The corresponding part of 
$-D_k(s)+\hat D_{k,o}(s)$ is not necessarily small, but
it is localized at the edge of $\Lambda_{n_k}$. Therefore, the commutator with
$\Pi_{L_k}\lmk \lmk \hat \alpha_s^{(k,o)}\rmk^{-1}\lmk A\rmk\rmk$
is small.
From the Lieb-Robinson bound (\ref{lr}), by the same kind of argument as in (\ref{two})
\begin{align}\label{four}
&\lV
\sum_{\substack{X\subset \Lambda_{n_k}\\
X\cap\Lambda_{m_k}^c\neq\emptyset
}}
\int_{|t|\le S_X^k} dt W_\gamma(t)
\left[ B(X,s,t,k),\;\Pi_{L_k}\lmk \lmk \hat \alpha_s^{(k,o)}\rmk^{-1}\lmk A\rmk\rmk
\right]
\rV\nonumber\\
&\le
\sum_{\substack{X\subset \Lambda_{n_k}\\
X\cap\Lambda_{m_k}^c\neq\emptyset
}}
\int_{|t|\le S_X^k} dt \lv W_\gamma(t)\rv
\lmk
\begin{gathered}
\lV
\left[\tau_{t}^{\Phi(s),\Lambda_{n_k}} \lmk \Phi'(X;s)\rmk
 ,\;\Pi_{L_k}\lmk \lmk \hat \alpha_s^{(k,o)}\rmk^{-1}\lmk A\rmk\rmk
\right]
\rV\\
+\lV
\left[\tau_{t}^{\Phi(s)+\Psi_{k}(s),\Lambda_{n_k}} \lmk \Phi'(X;s)\rmk ,\;\Pi_{L_k}\lmk \lmk \hat \alpha_s^{(k,o)}\rmk^{-1}\lmk A\rmk\rmk
\right]
\rV
\end{gathered}
\rmk\nonumber
\\
&\le
2C_{1,a}C_1\lV W_\gamma\rV_1\lV A\rV
\sum_{\substack{X\subset \Lambda_{n_k}\\
X\cap\Lambda_{m_k}^c\neq\emptyset\\\diam X<M
}}
e^{v_aS_X^k}
\sum_{x\in X,y\in \Lambda_{L_k}}F_a(|x-y|)\nonumber\\
&\le
2C_{1,a}C_1\lV W_\gamma\rV_1\lV A\rV M\sum_{y\in \bbZ}F(|y|)
\sum_{\substack{X\subset \Lambda_{n_k}\\
X\cap\Lambda_{m_k}^c\neq\emptyset\\\diam X<M
}}
e^{v_aS_X^k-a\cdot d\lmk X,\Lambda_{L_k}\rmk}\nonumber\\
&=
2C_{1,a}C_1\lV W_\gamma\rV_1\lV A\rV M\sum_{y\in \bbZ}F(|y|)
\sum_{j=m_k-M-L_k}^\infty 
\sum_{\substack{X\subset \Lambda_{n_k}\\
X\cap\Lambda_{m_k}^c\neq\emptyset\\\diam X<M\\
d\lmk
X,\Lambda_{L_k}
\rmk=j
}}
e^{v_aS_X^k-a\cdot d\lmk X,\Lambda_{L_k}\rmk}\nonumber\\
&\le
2C_{1,a}C_1\lV W_\gamma\rV_1\lV A\rV M\sum_{y\in \bbZ}F(|y|)2^{M}
\sum_{j=m_k-M-L_k}^\infty 
e^{-\frac{aj}2}.
\end{align}
As we assumed that $k$ is large enough so that $n_k\ge 4(M+R+l)$, we have $m_k-M-L_k\ge 1$.
Therefore, in the last inequality,
the number of $X\cap\lmk \Lambda_{m_k}\rmk^c\neq\emptyset$ 
with $\diam X<M$  and $d(X,\Lambda_{L_k})=j\ge 1$ is bounded by $2^{M}$.
The last line is independent of $s\in[0,1]$ and goes to $0$
as $k\to\infty$.

For $X\cap\lmk \Lambda_{m_k}\rmk^c\neq\emptyset$, and
$|t|\ge S_X^{k}$ part, 
we have
\begin{align}\label{three}
&\lV
\sum_{\substack{X\cap\lmk \Lambda_{m_k}\rmk^c\neq\emptyset\\X\subset\Lambda_{n_k}}}
\int_{|t|\ge S_X^k} dt W_\gamma(t) B(X,s,t,k)
\rV
\le 4C_1\sum_{\substack{X\cap\lmk \Lambda_{m_k}\rmk^c\neq\emptyset\\X\subset\Lambda_{n_k}\\
\diam X<M}} I_\gamma(S_{X}^k)
\nonumber\\
&\le
2^{M+2}C_1\sum_{j=m_k-M-L_k}^\infty I_\gamma\lmk \frac{aj}{2v_a}\rmk.
\end{align}
In the first inequality we used $B(X,s,t,k)\le 2C_1$ and (\ref{ig}) and the oddness of $W_\gamma(t)$.
As we assumed that $k$ is large enough so that $n_k\ge 4(M+R+l)$, we have $m_k-M-L_k\ge 1$.
Therefore, in the second inequality,
the number of $X\cap\lmk \Lambda_{m_k}\rmk^c\neq\emptyset$ 
with $\diam X<M$  and $d(X,\Lambda_{L_k})=j\ge 1$ is bounded by $2^{M}$.
The right hand side is independent of $s\in[0,1]$ and goes to $0$
as $k\to\infty$.

Similary, we may estimate  $X\subset \Lambda_{m_k}$
and $|t|\ge T_X^{k}$ part.  
\begin{align}\label{three}
&\lV
\sum_{X\subset \Lambda_{m_k}}
\int_{|t|\ge T_X^k} dt \; W_\gamma(t) B(X,s,t,k)
\rV
\le 4C_1\sum_{\substack{X\subset \Lambda_{m_k}\\\diam X<M}}I_\gamma(T_{X}^k)
\nonumber\\
&\le
2^{M+2}C_1\sum_{j=n_k-R-m_k}^\infty I_\gamma\lmk \frac{aj}{2v_a}\rmk.
\end{align}
The last line is independent of $s\in[0,1]$ and goes to $0$
as $k\to\infty$.

Hence we have shown
\begin{align}
\sup_{s\in[0,1]}
\lV
\left[
\sum_{X\subset \Lambda_{n_k}}
\int_{-\infty}^\infty dt W_\gamma(t)
B(X,s,t,k),\;
\Pi_{L_k}\lmk \lmk \hat \alpha_s^{(k,o)}\rmk^{-1}\lmk A\rmk\rmk
\right]
\rV\to 0,\quad k\to\infty.
\end{align}
The latter part of (\ref{bn}) can be estimated analogously.
We divide the integral into $|t|\le S_{X}^k$ part and $|t|\ge S_X^{k}$ part.
The $|t|\le S_{X}^k$ part can be treated as in (\ref{four}) and we have
\begin{align}
&\sum_{\substack{X\subset \Lambda_{n_k}\\X \subset\Lambda_{n_k}\setminus \Lambda_{n_k-R}}}
\int_{|t|\le S_{X}^k} dt \left|W_\gamma(t)\right|
\lV
\left[
\tau_{t}^{\Phi(s)+\Psi_{k}(s),\Lambda_{n_k}} \lmk \Psi_k'(X;s)\rmk,
\;\Pi_{L_k}\lmk \lmk \hat \alpha_s^{(k,o)}\rmk^{-1}\lmk A\rmk\rmk
\right]
\rV\nonumber\\
&\le
\sum_{\substack{X\subset \Lambda_{n_k}\\X \subset\Lambda_{n_k}\setminus \Lambda_{n_k-R}}}
C_1C_{1a}\lV W_\gamma\rV_1e^{v_aS_X^k-ad(X,\Lambda_{L_k})} (2R)\sum_{y\in\bbZ} F(|y|)\lV A\rV
\le C_1C_{1a}\lV W_\gamma\rV_1 2^{2R}(2R)\sum_{y\in\bbZ} F(|y|)
\sum_{l=(n_k-R-L_k)}^\infty e^{-\frac{al}2}\lV A\rV.
\end{align}
The last line is independent of $s\in[0,1]$ and goes to $0$
as $k\to\infty$.
The $|t|\ge S_{X}^k$ part can be treated as in (\ref{three}) and we have
\begin{align}
&\sum_{\substack{X\subset \Lambda_{n_k}\\X \subset\Lambda_{n_k}\setminus \Lambda_{n_k-R}}}
\int_{|t|\ge S_{X}^k} dt \;
\left|W_\gamma(t)\right|
\lV
\tau_{t}^{\Phi(s)+\Psi_{k}(s),\Lambda_{n_k}} \lmk \Psi_k'(X;s)\rmk
\rV\nonumber\\
&
\le
\sum_{\substack{X\subset \Lambda_{n_k}\\X \subset\Lambda_{n_k}\setminus \Lambda_{n_k-R}}}
2C_1I_\gamma(S_X^k)
\le 2^{2R+1}C_1\sum_{j=n_k-R-L_k}^\infty I_\gamma\lmk \frac{aj}{2v_a}\rmk.
\end{align}
The last line is independent of $s\in[0,1]$ and goes to $0$
as $k\to\infty$.

Hence we have shown
\begin{align}\label{zenhan}
\sup_{s\in[0,1]}\lV \left[
-D_k(s)+\hat D_{k,o}(s), \;\Pi_{L_k}\lmk \lmk \hat \alpha_s^{(k,o)}\rmk^{-1}\lmk A\rmk\rmk
\right]\rV
\to 0.
\end{align}

We also bound $-D_k(s)+\hat D_{k,o}(s)$ itself.
From (\ref{ppp})
\begin{align}
&\lV-D_k(s)+\hat D_{k,o}(s)\rV\nonumber\\
&\le
\sum_{\substack{X\subset \Lambda_{n_k}\\\diam X<M}}
 \int_{|t|\le T_X^k} dt |W_\gamma(t)|C_1C_{3,a}\sum_{x\in X}e^{v_a|t|-a\cdot d(x,\lmk \Lambda_{n_k-R}\rmk^c)}
 +2C_1\sum_{\substack{X\subset \Lambda_{n_k}\\ \diam X<M}}\int_{|t|\ge T_X^k} dt |W_\gamma(t)|\nonumber\\
& +\sum_{X \subset\Lambda_{n_k}\setminus \Lambda_{n_k-R}}
C_1\int_{-\infty}^\infty dt |W_\gamma(t)|
\nonumber\\
&\le\lmk
 C_1C_{3,a}M\lV W_\gamma\rV_1
\sum_{l=1}^\infty
\sum_{\substack{X\subset \Lambda_{n_k}\\\diam X<M\\ d(X,\lmk \Lambda_{n_k-R}\rmk^c)=l}}
 e^{-\frac{al}2}
 \rmk
 +\lmk
 4C_1
 \sum_{l=0}^\infty I_\gamma\lmk \frac{al}{2v_a}\rmk
 \sum_{\substack{X\subset \Lambda_{n_k}\\ \diam X<M\\ d(X,\lmk \Lambda_{n_k-R}\rmk^c)=l}}
 1
 \rmk
  +
  \lmk
  2^{2R}C_1
\lV W_\gamma\rV_1\rmk
\nonumber\\
&\le \lmk
C_1C_{3,a}2^{M}M\lV W_\gamma\rV_1
\sum_{l=1}^\infty
 e^{-\frac{al}2}
 \rmk
 +\lmk
 2^{M+3}C_1R
 \sum_{l=0}^\infty I_\gamma\lmk \frac{al}{2v_a}\rmk\rmk
 +2^{2R}C_1
\lV W_\gamma\rV_1
\end{align}
In the second inequality, we used the fact that $T_X^k=0$ if $d(X,\lmk \Lambda_{n_k-R}\rmk^c))=0$.
The last line is finite and independent of $s\in[0,1]$ and $k\in\nan$.
Combining this with (\ref{dfr}),
we obtain
\begin{align}\label{kouhan}
\sup_{s\in[0,1]}\lV\left[
-D_k(s)+\hat D_{k,o}(s),\;
\lmk \hat \alpha_s^{(k,o)}\rmk^{-1}\lmk A\rmk-
\Pi_{L_k}\lmk \lmk \hat \alpha_s^{(k,o)}\rmk^{-1}\lmk A\rmk\rmk
\right ]\rV
\to
0,\quad k\to\infty.
\end{align}
From (\ref{zenhan}) and (\ref{kouhan}), we obtain (\ref{claim}).

From (\ref{claim}),
we prove (\ref{ibng}), 
\begin{align}
&\lV 
\lmk \alpha_s^{(k)}\rmk^{-1}\lmk A\rmk-\lmk \hat \alpha_{s}^{(k,o)}\rmk^{-1}(A)
\rV
=\lV 
A-\alpha_s^{(k)}\circ \lmk \hat \alpha_{s}^{(k,o)}\rmk^{-1}(A)
\rV
=\lV
\int_0^s du\;
\frac{d}{du}\alpha_u^{(k)}\circ \lmk \hat \alpha_{u}^{(k,o)}\rmk^{-1}(A)
\rV\nonumber\\
&=\lV
\int_0^s du\;
\alpha_u^{(k)}\lmk
i\left[-D_k(u)+\hat D_{k,o}(u),\;
\lmk \hat \alpha_{u}^{(k,o)}\rmk^{-1}(A)
\right]
\rmk
\rV
\le \varepsilon_k(A)\to 0,\quad k\to\infty,
\end{align}
for any $l\in\nan$ and $A\in\caA_{\Lambda_l}$.
Hence we have
\begin{align}
\lim_{k\to\infty}\lV 
\lmk \alpha_s^{(k)}\rmk^{-1}\lmk A\rmk-\lmk \hat \alpha_{s}^{(k,o)}\rmk^{-1}(A)
\rV=0,
\end{align}
 for any $A\in\caA$.
As we also have
\begin{align}
\lV 
\lmk \hat \alpha_{s}^{(k,o)}\rmk^{-1}(A)-\lmk \alpha_{s,o}\rmk^{-1}(A)
\rV\to 0,\quad k\to \infty,
\end{align} 
for any $A\in\caA$
from \cite{bmns},
we obtain
\begin{align}
\lV 
\lmk \alpha_s^{(k)}\rmk^{-1}\lmk A\rmk-\lmk \alpha_{s,o}\rmk^{-1}(A)
\rV\to 0,\quad k\to \infty,
\end{align}
for any $A\in\caA$.
From this, we have
\begin{align}
\lV
\alpha_{s,o}(A)-\alpha_s^{(k)}(A)
\rV
=
\lV
\alpha_s^{(k)}\lmk
\lmk\alpha_s^{(k)}\rmk^{-1}- \lmk \alpha_{s,o}\rmk^{-1}
\rmk
\alpha_{s,o}(A)
\rV
\to 0,\quad k\to\infty,
\end{align}
for any $A\in\caA$.
Hence we have proven the Lemma.

\end{proofof}

\begin{proofof}[Lemma \ref{vvvk}]

First we prove
\begin{align}\label{ni}
\sup_{s\in[0,1]}\lmk
\sum_{X\subset\Lambda_{m_k}}\int_{-\infty}^\infty dt\; \lv W_\gamma(t)\rv
\lV
-\tau_t^{\Phi(s),\Lambda_{n_k}}\lmk\Phi'(X;s)\rmk
+\tau_t^{\Phi(s)}\lmk\Phi'(X;s)\rmk
\rV
\rmk
\to 0,\quad
k\to\infty.
\end{align}
To prove this, for each $X\in{\mathfrak S}_{\bbZ}$ and $k\in\nan$
we set
\begin{align}\label{sx}S_X^{(k)}:=\frac{a}{2v_a}d(\Lambda_{n_k}^c,X).
\end{align}
With this $S_X^{(k)}$, we divide the integral into $|t|\le S_X^{(k)}$ part and
$|t|\ge S_X^{(k)}$ part.
By (\ref{nos}) and Definition \ref{boundary} {\it 2.},
 $|t|\le S_X^{(k)}$ part is bounded as
\begin{align}
&\sum_{X\subset\Lambda_{m_k}}\int_{|t|\le S_X^{(k)}} dt\; \lv W_\gamma(t)\rv
\lV
-\tau_t^{\Phi(s),\Lambda_{n_k}}\lmk\Phi'(X;s)\rmk
+\tau_t^{\Phi(s)}\lmk\Phi'(X;s)\rmk
\rV
\le 
C_{4,a}
e^{-\frac{a}{2}(n_k-m_k)}.
\end{align}
Here $C_{4,a}$ is a positive constant which is independent of $k,s$.
The right hand side is indepenednt of $s\in[0,1]$ and converges to $0$ as $k\to\infty$.
The $|t|\ge S_X^{(k)}$ part 
\begin{align}\label{sxi}
&\sum_{X\subset\Lambda_{m_k}}\int_{|t|\ge S_X^{(k)}} dt\; \lv W_\gamma(t)\rv
\lV
-\tau_t^{\Phi(s),\Lambda_{n_k}}\lmk\Phi'(X;s)\rmk
+\tau_t^{\Phi(s)}\lmk\Phi'(X;s)\rmk
\rV\nonumber\\
&\le 
2\sum_{X\subset\Lambda_{m_k}}\int_{|t|\ge S_X^{(k)}} dt\; \lv W_\gamma(t)\rv
\lV
\Phi'(X;s)
\rV
\le
4C_1\sum_{\substack{X\subset\Lambda_{m_k}\\
\diam(X)<M}}I_\gamma(S_X^{(k)})\nonumber\\
&=4C_1\sum_{l=n_k-m_k}^\infty\sum_{\substack{X\subset\Lambda_{m_k}\\
\diam(X)<M\\d(X,\Lambda_{n_k}^c)=l}
}I_\gamma(S_X^{(k)})
\le
4C_1\cdot 2^{M}\sum_{l=n_k-m_k}^\infty
I_\gamma(\frac{a}{2v_a}l).
\end{align}
Here, we used Definition \ref{boundary} {\it 2.} for the second inequality.
In the third line, we recalled the definition of $S_X^{(k)}$ (\ref{sx}) and
used the fact that for any finite set $X$ in $\Lambda_{m_k}$
with $\diam(X)<M$, the distance between $X$ and $\Lambda_{n_k}^c$
is at least $n_k-m_k$.
We also used the fact 
that for any $l\ge n_k-m_k$,
the number of $X\subset\Lambda_{m_k}$ with $\diam(X)<M$
such that $d(X,\Lambda_{n_k}^c)=l$
is at most $2^{M}$.
The right hand side of (\ref{sxi}) is independent of
$s\in[0,1]$ goes to $0$
as $k\to\infty$, because of (\ref{wg}).
Hence we have shown (\ref{ni}).
Similarly,
we have
\begin{align}\label{ichi}
\sup_{s\in[0,1]}\lmk
\sum_{X\subset\Lambda_{m_k}}\int_{-\infty}^\infty dt\; \lv W_\gamma(t)\rv
\lV
-\tau_t^{\tilde\Phi(s),\Lambda_{n_k}}\lmk\tilde\Phi'(X;s)\rmk
+\tau_t^{\tilde\Phi(s)}\lmk\tilde\Phi'(X;s)\rmk
\rV
\rmk
\to 0,\quad
k\to\infty.
\end{align}
Next we show
\begin{align}\label{104}
\sup_{s\in[0,1]}\lmk
\int_{-\infty}^\infty dt\; \lv W_\gamma(t)\rv
\sum_{\substack{X\in {\mathfrak S}_\bbZ\\
X\cap \Lambda_{m_k}^c\neq \emptyset}}
\lV
\tau_t^{\tilde\Phi(s)}\lmk\tilde\Phi'(X;s)\rmk
-\tau_t^{ \Phi(s)}\lmk\Phi'(X;s)\rmk
\rV
\rmk
\to 0,\quad
k\to\infty.
\end{align}
To prove this, for each $X\in{\mathfrak S}_{\bbZ}$,
we set
\begin{align}\label{tx}
&R_X:=\min\left\{
d(X,Y)\mid Y\cap[0,\infty)\neq\emptyset,\; 
Y\cap(-\infty,-1]\neq\emptyset,\;\diam Y<M\right\},\\
&T_X:=\frac{a}{2v_a}R_X.
\end{align}
With this $T_X$, we divide the integral into $|t|\le T_X$ part and
$|t|\ge T_X$ part.
We then have
\begin{align}
&\sum_{\substack{X\in {\mathfrak S}_\bbZ\\
X\cap \Lambda_{m_k}^c\neq \emptyset}}
\int_{-\infty}^\infty dt\; 
\lv W_\gamma(t)\rv
\lV
\tau_t^{\tilde\Phi(s)}\lmk\tilde\Phi'(X;s)\rmk
-\tau_t^{ \Phi(s)}\lmk\Phi'(X;s)\rmk
\rV\\
&\le
\sum_{\substack{X\in {\mathfrak S}_\bbZ\\
X\cap \Lambda_{m_k}^c\neq \emptyset}}
\int_{|t|\le T_X} dt\; 
\lv W_\gamma(t)\rv
\lV
\tau_t^{\tilde\Phi(s)}\lmk\Phi'(X;s)\rmk
-\tau_t^{ \Phi(s)}\lmk\Phi'(X;s)\rmk
\rV\label{7lb}\\
&+
\sum_{\substack{X\in {\mathfrak S}_\bbZ\\
X\cap \Lambda_{m_k}^c\neq \emptyset}}
\int_{|t|\le T_X} dt\; 
\lv W_\gamma(t)\rv
\lV
\tau_t^{\tilde\Phi(s)}\lmk\tilde\Phi'(X;s)-\Phi'(X;s)\rmk
\rV\label{72b}\\
&+
\sum_{\substack{X\in {\mathfrak S}_\bbZ\\
X\cap \Lambda_{m_k}^c\neq \emptyset}}
\int_{|t|\ge T_X} dt\; 
\lv W_\gamma(t)\rv
\lV
\tau_t^{\tilde\Phi(s)}\lmk\tilde\Phi'(X;s)\rmk
-\tau_t^{ \Phi(s)}\lmk\Phi'(X;s)\rmk
\rV.\label{73b}
\end{align}
The first part (\ref{7lb}) is bounded by use of the (\ref{tppd}) as
\begin{align*}
|(\ref{7lb})|
\le
\lV W_\gamma\rV_1 C_1C_{2,a}M
\sum_{l=m_k-M}^\infty
\sum_{\substack{X\in {\mathfrak S}_\bbZ\\
X\cap \Lambda_{m_k}^c\neq \emptyset\\
\diam X<M\\
d(X,\{0\})=l}}
e^{\frac a 2 (-l+M)}
\le C_{5,a}\sum_{l=m_k-M}^\infty e^{-\frac{al}2}.
\end{align*}
In the last line, we used 
$R_X\le d(X,\{0,-1\})\le d(\{x\},\{0\})$ for all $x\in X$ and 
$d(X,\{0\})-M\le d(X,[-M,M])\le R_X$. 
(Recall we assumed $M>2$ in the beginning of this section.)
We also used
the fact that 
the number of $X$ with
$\diam X<M$ 
and $d(X,\{0\})=l$ is bounded by $2^{M}$,
and introduced a new constant $C_{5,a}:=2^{M}M \lV W_\gamma\rV_1 C_1C_{2,a}e^{\frac{a}2{M}}$.
The right hand side is independent of $s\in[0,1]$ and goes to $0$
as $k\to\infty$.
The second term (\ref{72b}) is $0$ for $k$ large enough.
The third term (\ref{73b}) can be evaluated as in (\ref{sxi}).
We have for $m_k>2M$,
\begin{align}
|(\ref{73b})|
\le
4C_1
\sum_{\substack{X\in {\mathfrak S}_\bbZ\\
X\cap \Lambda_{m_k}^c\neq \emptyset\\\diam X<M}}
I_\gamma(T_X)
\le
4C_1
\sum_{l=m_k-M}^\infty
\sum_{\substack{X\in {\mathfrak S}_\bbZ\\
X\cap \Lambda_{m_k}^c\neq \emptyset\\\diam X<M\\
d(X,\{0\})=l}}I_\gamma\lmk \frac{a}{2v_a}(l-M)\rmk
\le 4C_1 2^{M}
\sum_{l=m_k-2M}^\infty
I_\gamma\lmk \frac{a}{2v_a}l\rmk.
\end{align}
Here we used 
$d(X,\{0\})-M\le R_X$, for the second inequality.

The right hand side is independent of $s\in[0,1]$ and goes to $0$
as $k\to\infty$.
Hence we have shown (\ref{104}).
Similarly, we obtain
\begin{align}\label{san}
\sup_{s\in[0,1]}\lmk
\int_{-\infty}^\infty dt\; \lv W_\gamma(t)\rv
\sum_{\substack{X\subset\Lambda_{n_k}\\
 X\cap \Lambda_{m_k}^c\neq \emptyset}}
\lV
\tau_t^{\tilde\Phi(s),\Lambda_{n_k}}\lmk\tilde\Phi'(X,s)\rmk
-\tau_t^{\Phi(s),\Lambda_{n_k}}\lmk\Phi'(X,s)\rmk
\rV
\rmk
\to 0,\quad
k\to\infty.
\end{align}

From (\ref{104}),
we have
\begin{align}
\int_{-\infty}^\infty dt\;
\lv W_\gamma(t)\rv\;\;
\lmk
\sum_{X\in{\mathfrak S}_\bbZ}
\lV
\tau_t^{\tilde\Phi(s)} \lmk \tilde \Phi'(X;s)\rmk
- \tau_t^{\Phi(s)} \lmk  \Phi'(X;s)\rmk
\rV
\rmk<\infty.
\end{align}
Therefore, we may define
\begin{align}
V(s):=
\int_{-\infty}^\infty dt\;
 W_\gamma(t)\;\;
\lmk
\sum_{X\in{\mathfrak S}_\bbZ}
 \tau_t^{\tilde\Phi(s)} \lmk \tilde \Phi'(X;s)\rmk
- \tau_t^{\Phi(s)} \lmk  \Phi'(X;s)\rmk
\rmk\in\caA,
\end{align}
and from (\ref{ni}), (\ref{ichi}), (\ref{104}), (\ref{san}),
we obtain (\ref{vcon}).
\end{proofof}

\section{On-site group symmetry}\label{onsite}
For a Hilbert space $\caK$, we denote by $\caU(\caK)$ the set of all unitaries on $\caK$.
Let $G$ be a {\new {finite}} group and $w:G\to \caU(\bbC^{2S+1})$ a unitary representation of
$G$ on $\bbC^{2S+1}$.
Then there is an action $T:G\to \Aut \caA$ of $G$ on $\caA$ 
such that
\begin{align}
T_g\lmk
A
\rmk
=\lmk
\bigotimes_I w(g)\rmk
A
\lmk
\bigotimes_I w(g)^*\rmk,\quad g\in G,\quad A\in\caA_I,
\end{align}
for any finite interval $I$ of $\bbZ$.
A state $\varphi$ on $\caA$ is $G$-invariant if $\varphi\circ T_g=\varphi$ for any $g\in G$.
As $T_g (\caA_R)=\caA_R$, the restriction $T_{g,R}:=\left. T_g\right \vert_{\caA_R}$
is a $*$-automorphism on $\caA_R$. 

In \cite{Matsui1}, Matsui introduced the projective representation of $G$ associated 
to pure split $G$-invariant states.
As in Theorem \ref{sp}, it is unique up to unitary conjugacy and a phase, and the cohomology class is independent of the choice of the projective representation.
\begin{thm}
Let $\varphi$ be a $G$-invariant pure state on $\caA$, which satisfies the split property. Let $\varphi_R$ be the restriciton of
$\varphi$ to $\caA_R$, and $(\caH_{\varphi_R},\pi_{\varphi_R},\Omega_{\varphi_R})$ be the GNS triple of $\varphi_R$.
Then there are a Hilbert space $\caK_\varphi$, a $*$-isomorphism $\iota_\varphi : \pi_{\varphi_R}\lmk \caA_R\rmk{''}\to B(\caK_{\varphi})$,
and a projective unitary representation $U_\varphi:G\to\caU(\caK_\varphi)$ on $\caK_{\varphi}$ such that
\begin{align*}
\iota_\varphi\circ \pi_{\varphi_R}\circ T_{g,R}\lmk A\rmk
=U_\varphi(g) \lmk \iota_\varphi \circ\pi_{\varphi_R}\lmk A\rmk\rmk U_\varphi(g)^*,\quad
A\in\caA_R,\quad g\in G.
\end{align*}
These $\caK_\varphi$, $\iota_\varphi$, $U_{\varphi}$ are unique in the following sense.: 
If  a Hilbert space  $\tilde \caK_\varphi$, a $*$-isomorphism $\tilde\iota_\varphi : \pi_{\varphi_R}\lmk \caA_R\rmk{''}\to B(\tilde \caK_{\varphi})$,
and a projective unitary representation $\tilde U_\varphi:G\to\caU(\tilde \caK_\varphi)$ on $\tilde\caK_{\varphi}$ satisfy
\begin{align*}
\tilde \iota_\varphi\circ \pi_{\varphi_R}\circ T_{g,R}\lmk A\rmk
=\tilde U_\varphi(g) \lmk \tilde \iota_\varphi \circ\pi_{\varphi_R}\lmk A\rmk\rmk {\tilde U_\varphi(g)}^*,\quad
A\in\caA_R,\quad g\in G,
\end{align*}
then
there is a unitary $W:\caK_\varphi\to \tilde \caK_\varphi$ and $c: G\to  \bbT$
such that
\begin{align*}
&W\lmk \iota_\varphi \lmk x\rmk\rmk W^*=
\tilde \iota_\varphi \lmk x\rmk,
 \quad x\in \pi_{\varphi_R}\lmk \caA_R\rmk{''},\\
& c(g) WU_\varphi(g) W^*=\tilde U_\varphi(g),\quad g\in G.
 \end{align*}
In particular, the cohomology class of $U_\varphi$ is equal to that of $\tilde U_{\varphi}$.
\end{thm}
The same argument as the proof of Theorem \ref{aei} shows that the cohomology class is an invariant
of
factorizable automorphic equivalence, preserving $G$-symmetry. 
\begin{thm}
Let $\varphi_1,\varphi_2$ be $G$-invariant pure states satisfying the split property.
Suppose that there 
 exists an automorphism $\alpha$ on $\caA$ such that
\begin{align}
\varphi_2=\varphi_1\circ \alpha\quad\text{and}\quad
\alpha\circ T_g=T_g\circ\alpha,\quad g\in G.
\end{align}
Furthermore, assume that there are automorphisms
$\alpha_R$,
$\alpha_L$ on $\caA_R$, $\caA_L$ respectively, and a unitary $W$ in $\caA$
such that
\begin{align}
\alpha_R\circ T_{g,R}=T_{g,R}\circ\alpha_R,\quad g\in G
\end{align}
and
\[
\alpha\circ\lmk\alpha_L^{-1}\otimes \alpha_R^{-1}\rmk(A)=
WAW^*,\quad A\in\caA.
\]
Then the 
 the cohomology class of the associated projective representations of $\varphi_1$ and 
 $\varphi_2$ are equal.
 \end{thm}
From this, we can show that the cohomology class is invariant of $C^1$-classification.
\begin{thm}
Let $\Phi:[0,1]\ni s \to \Phi(s):=\{\Phi(X;s)\}_{X\in {\mathfrak S}_\bbZ}\in {\caB_f}$ be a $C^1$-path of
interactions, satisfying the {\it Condition B} 
with
\begin{center}
{\it 6'.} For each $s\in[0,1]$, $\Phi(s)$ is $G$-invariant i.e.,
\[
T_g\lmk \Phi(X;s)\rmk=\Phi(X;s),\quad g\in G,\quad X\in\mathfrak S_\bbZ,
\]
\end{center}
instead of {\it 6}.
Then
the cohomology class of the associated representation of the ground state
 does not change along the path.
\end{thm}

\end{document}